\documentclass[11pt,draftcls,onecolumn,twoside]{IEEEtran}
\usepackage[latin1]{inputenc}
\usepackage[mathcal]{euscript}
\usepackage{amssymb,amsmath,amsthm}
\usepackage{graphicx}
\usepackage[english]{babel}
\usepackage{textcomp}
\usepackage{setspace}
\usepackage{mathtools}
\usepackage{algorithm}
\usepackage{algorithmic}

\newtheorem{remark}{Remark}

\newtheorem{problem}{Problem}
\newtheorem{theo}{Theorem}
\newtheorem{lemma}{Lemma}

\newtheorem{claim}{Claim}
\newtheorem{assumption}{Assumption}

\usepackage{color}
\definecolor{MyDarkBlue}{rgb}{0.1,0,0.65}
\definecolor{MyDarkRed}{rgb}{.85,0,0.1}

\newcommand{\card}[1]{\textrm{card}\left(#1\right)}
\newcommand{\EE}[1]{{\mathbb{E}}\left[ #1 \right]}

\newcommand{\dsp}{\displaystyle}

\newcommand{\ki}{_{k,i}}
\newcommand{\ku}{_{k,1}}
\newcommand{\kd}{_{k,2}}

\newcommand{\K}{^{(K)}}
\newcommand{\cK}{^{c,(K)}}

\newcommand{\SinfTwo}
{\eqref{eq:Q1d_2}-\eqref{eq:gamma1d_2}-\eqref{eq:gamma2d_2}}

\title{Nearly Optimal Resource Allocation for Downlink OFDMA in 2-D
Cellular Networks}

\author{Nassar Ksairi, Pascal Bianchi and Philippe Ciblat}

\begin{document}
\maketitle

\begin{abstract}
In this paper, we propose a resource allocation algorithm for the downlink of
sectorized two-dimensional (2-D) OFDMA cellular networks assuming statistical
Channel State Information (CSI) and fractional frequency reuse. 
The proposed algorithm can be implemented in a distributed fashion
without the need to any central controlling units. Its performance is analyzed
assuming fast fading Rayleigh channels and Gaussian distributed multicell
interference. We show that the transmit power of this simple algorithm tends, as
the number of users grows to infinity, to the same limit as the minimal power
required to satisfy all users' rate requirements \emph{i.e.,} the proposed
resource allocation algorithm is asymptotically optimal. As a byproduct of this
asymptotic analysis, we characterize a relevant value of the reuse factor that
only depends on an average state of the network. 
\end{abstract}


\section{Introduction}

We address in this work the problem of resource allocation (power control and
subcarrier assignment) for the downlink of sectorized OFDMA networks impaired
with multicell interference. A considerable research interest has been lately
dedicated to this problem since the adoption of OFDMA in a number of current and
future wireless standards such as WiMax and 3GPP-LTE. In principle, the problem
of resource allocation should be jointly solved in all the cells of the system.
In most of the practical situations, this optimization problem is difficult to
solve. Therefore, most of the related works in the literature focus on the
single cell case
(\emph{e.g.,}~\cite{ofdma_duality_gap}-\cite{ofdma_individual1}). 
Fewer works address the more involved multicell allocation problem. In this
context, we cite~\cite{scaling_laws}-\cite{imp_reuse} in the case of perfect CSI
at the transmitters side, and~\cite{pischella_mimo_ofdma,gau-hac-cib-1} in the
case of imperfect CSI. In~\cite{pischella_mimo_ofdma,gau-hac-cib-1}, \emph{all}
the available subcarriers are likely to be used by different base stations and
are thus subject to multicell interference. In such a configuration,
interference may reach excessive levels, especially for users located at cells
borders.

Similarly to~\cite{pischella_mimo_ofdma,gau-hac-cib-1}, we assume in this paper
that users' channels undergo fast fading and that the CSI at the base stations
is limited to some channel statistics. However, contrary to these two works, we
consider that a certain subset of subcarriers is shared orthogonally between the
adjacent base stations (and is thus ``protected'' from multicell interference)
while the remaining subcarriers are ``non protected'' since they are reused by
different base stations. This so-called \emph{fractional frequency reuse} (or
FFR) is recommended in a number of standards \emph{e.g.,} in~\cite{wimax2} for
IEEE 802.16 (WiMax)~\cite{wimax1}, as a way to avoid severe inter-cell
interference. The ratio between the number of non protected subcarriers and the
total number of subcarriers is generally referred to as the \emph{reuse factor}
and is denoted in the sequel by~$\alpha$.

Few works in the literature (we cite~\cite{ffr_mimo,ffr_self,ffr_suppression}
without being exclusive) have addressed the problem of resource allocation for
FFR-based OFDMA networks, and none of them fits into the above framework which
is considered in this paper. The particular problem considered
in~\cite{ffr_mimo} consists in maximizing a system-wide utility function under a
power constraint. In this context, the authors propose a distributed iterative
allocation algorithm that is based on estimating the level of multicell
interference rather than computing it. Of course, resource allocation schemes
that do not resort to such simplifications are highly preferable. In the same
context, authors of~\cite{ffr_self} consider the problem of minimizing the total
transmit power needed to satisfy all users' rate requirements. For that sake,
they propose a heuristic allocation algorithm without any assessment of its
deviation from the optimal solution to the latter problem. Moreover, the
selection of a relevant reuse factor is not addressed. Finally, authors
of~\cite{ffr_suppression} assume that subcarrier assignment is done separately
and in advance. The major drawback of this work is thus that joint power control
and subcarrier assignment is not addressed.

In our work, we investigate the problem of power control and subcarrier
assignment for the downlink of FFR-based OFDMA systems allowing to satisfy all
users' rate requirements while spending the least possible power at the
transmitters' side. In our previous work~\cite{tsp1,tsp2}, the solution to this
problem is characterized in the special case of one-dimensional (1-D) cellular
networks where all users and base stations are located on a line. Unfortunately,
it is much more difficult to characterize this solution in the case of 2-D
networks. In the present work, our aim is to propose a suboptimal resource
allocation strategy for these 2-D networks and to study its performance with
respect to the above optimization problem. Our allocation algorithm assumes that
users of each cell are divided prior to resource allocation into two groups
separated by a fixed curve. The first group is composed of closer users to the
base station. These users are constrained to modulate non protected subcarriers
and are thus subject to multicell interference. The second group comprises the
farthest users who are constrained to modulate interference-free subcarriers. In
order to relevantly select the aforementioned separating curves, we study the
limit of an {\bf optimal} solution to the resource allocation problem as the
number of users grows to infinity. We show that if the curves are set using the
results of this asymptotic analysis, then the limit of the transmit power of the
proposed \emph{suboptimal} algorithm is equal to the limit of the transmit power
of the \emph{optimal} resource allocation. As a byproduct, we are able to
determine a relevant value of the reuse factor. Indeed, the asymptotic transmit
power depends on the average rate requirement and on the density of users in
each cell. It also depends on the value~$\alpha$ of the frequency reuse factor.
We can therefore define the optimal reuse factor as the value of~$\alpha$ which
minimizes this asymptotic power. The main contributions of this work are thus
the following.
\begin{enumerate}
\item  A practical resource allocation algorithm that can be implemented in a
distributed manner is proposed for the downlink of a sectorized OFDMA network
assuming fractional frequency reuse and statistical CSI. The transmit power of
this simple algorithm tends, as the number of users grows to infinity, to the
same limit as the minimal power required to satisfy all users' rate
requirements.
\item As a byproduct of our study of the above algorithm, we prove that the
simple scheme consisting in separating users of each cell beforehand into
interference-free users (constrained to modulate only non reusable subcarriers)
and interference users (constrained to modulate only reusable subcarriers) is
asymptotically optimal. This scheme is frequently used in cellular systems, but
it has never been proved optimal in any sense to the best of our knowledge.
\item Finally, a method is proposed to select a relevant value of the reuse
factor. The determination of this factor is of great importance for the
dimensioning of wireless networks. 
\end{enumerate}

The rest of this paper is organized as follows. The system model is introduced
in Section~\ref{sec:sys_model}, followed by a description of the multicell
resource allocation problem in Section~\ref{sec:joint_2d}. The proposed resource
allocation algorithm is presented in Section~\ref{sec:sub_optimal_2d}. The
relevant choice of the curves associated with this algorithm and which separate
the two groups of users in each cell is addressed in
Section~\ref{sec:asym_analysis_2d}. Next, the relevant selection of the reuse
factor is addressed in Subsection~\ref{sec:alpha}. Finally, The relevancy of the
proposed resource allocation and of our selection of the reuse factor are
sustained by simulations in Section~\ref{sec:simus}.

\section{System Model}
\label{sec:sys_model}

Consider the downlink of a sectorized OFDMA cellular network composed of
hexagonal cells. Each cell in the system is divided into three $120\,^{\circ}$
sectors. In this paper, we restrict ourselves to the case of three interfering
sectors of three adjacent cells, say cells~$A,B,C$ (see
Figure~\ref{fig:3cells_2D_model}). In the more general case of networks with
more than three cells, our results hold provided that the interference generated
by farther base stations can be neglected. Generally, this assumption is only
valid as a first approximation. However, it allows for an essential reduction of
the dimensionality of the multicell resource allocation problem. In the sequel,
we assume that the considered sectors of cells~$A,B,C$ have the same surface and
we denote by $K^A,K^B,K^C$ their respective number of users. Let
$K=K^A+K^B+K^C$~be the total number of users and $N$ the total number of
available subcarriers. The signal received by user~$k$ in cell~$c$
($c\in\{A,B,C\}$) at subcarrier~$n\in\{0 \ldots N-1\}$ during the $m$th
OFDM block is given by \begin{equation}\label{eq:signal_model}
y_k(n,m) = H_k^c(n,m) s_k(n,m) + w_k(n,m)\:,
\end{equation} 
where $s_k(n,m)$ represents the data symbol destined to user~$k$, and where
$w_k(n,m)$ is a random process that encompasses both the thermal noise and
the possible multicell interference. Random variable $H_k^c(n,m)$ stands for the
frequency-domain channel coefficient associated with user~$k$ in cell~$c$ at the
$n$th subcarrier and the $m$th OFDM block. The realizations of this random
variable are assumed to be known only at the receiver side and unknown at the
base station. Random variables $\{H_k^c(n,m)\}_{n,m}$ are Rayleigh distributed
with variance $\rho_k=\mathbb{E}[|H_k^c(m,n)|^2]$ which is assumed to be
constant w.r.t~$n$ and~$m$. This holds for example in the case of uncorrelated
time-domain channel coefficients. Furthermore, for each $n\in\{0
\ldots N-1\}$, random process $\{H_k^c$ $(n$, $m)\}_m$ is assumed to be ergodic.
Finally, variance $\rho_k$ is assumed to be known at the transmitter side and
{\bf vanishes with the distance between base station~$c$ and user~$k$ following
a given path loss model}. We assume that fractional frequency reuse is applied.
According to this scheme (see Figure~\ref{fig:3cells_2D_model}), a certain
subset of subcarriers $\mathcal{I} \subset \{0 \ldots N -1\}$ is reused in
the three cells. If user~$k$ modulates a subcarrier $n \in \mathcal{I}$, the
noise $w_k(n,m)$ includes both thermal noise and multicell interference.
The reuse factor $\alpha$ is the ratio between the number of reused subcarriers
and the total number of subcarriers:
\begin{equation*}
\alpha=\frac{\mbox{card}(\mathcal{I})}{N}\:.
\end{equation*}
The remaining $(1-\alpha)N$ subcarriers are shared by the three sectors in an
orthogonal way, such that each base station~$c$ ($c=A,B,C$) has at its disposal
a subset $\mathcal{P}_c$ of cardinality $\frac{1-\alpha}{3}N$. If user~$k$
modulates a subcarrier $n \in \mathcal{P}_c$, process $w_k(n,m)$ will contain
only thermal noise with variance $\sigma^2$. Finally, 
$\mathcal{I} \cup \mathcal{P}_A \cup \mathcal{P}_B \cup \mathcal{P}_C = 
\{0, 1, \ldots ,N-1\}$. Denote by $\mathcal{N}_k$ the subset of subcarriers
assigned to user~$k$. We assume that $\mathcal{N}_k$ may contain subcarriers
from both the ``interference'' subset $\mathcal{I}$ and the ``protected'' subset
$\mathcal{P}_c$. Denote by $\gamma\ku^c N$ (resp. $\gamma\kd^c N$) the number
of subcarriers assigned to user~$k$ in $\mathcal{I}$ (resp. $\mathcal{P}_c$).
In other words,
\begin{equation*}
\gamma\ku^c = \mbox{card}(\mathcal{I} \cap \mathcal{N}_k)/N \qquad 
\gamma\kd^c = \mbox{card}(\mathcal{P}_c \cap \mathcal{N}_k)/N\:.
\end{equation*}
Parameters $\gamma\ku^c$ and $\gamma\kd^c$ are generally referred to as
\emph{sharing factors}. We assume from now on that they can take on any value in
the interval $[0,1]$ (not necessarily integer multiples of $1/N$).
\begin{remark}
Even when the sharing factors are not integer multiples of $1/N$, it is still
possible to practically achieve the exact values of $\gamma\ku^c$, $\gamma\kd^c$
by simply exploiting the time dimension. Indeed, the number of subcarriers
assigned to user~$k$ can be chosen to vary from one OFDM symbol to another in
such a way that the \emph{average} number of subcarriers in subsets
$\mathcal{I}$ and $\mathcal{P}_c$ is equal to $\gamma\ku^c N$ and
$\gamma\kd^c N$ respectively. Thus the fact that $\gamma\ku^c$, $\gamma\kd^c$
are not strictly integer multiples of $1/N$ is not restrictive, provided that
the system is able to grasp the benefits of the time dimension. The particular
case where the number of subcarriers is restricted to be the same in each OFDM
block is addressed in Section~\ref{sec:simus}.
\end{remark}
Note that by definition
$$
\sum_{k=1}^{K^c}\gamma\ku^c \leq \alpha\:, \qquad \sum_{k=1}^{K^c}\gamma\kd^c
\leq \frac{1-\alpha}{3}\:.
$$ 
For the sake of readability and compactness of the paper, the above two
inequality constraints will be written from now on as equalities \emph{i.e.,} we
force the whole set of available subcarriers to be fully occupied by setting
$\sum_{k=1}^{K^c}\gamma\ku^c = \alpha$ and 
$\sum_{k=1}^{K^c}\gamma\kd^c = \frac{1-\alpha}{3}$. Indeed, keeping the above
constraints as inequalities would make the presentation of the final results as
well as of the proofs very tedious.

Recall that in our model, for each user~$k$ in any cell~$c$, all channel
coefficients $H_k^c(n,m)$ are identically distributed on all the subcarriers
assigned to this user (the variance $\rho_k=\mathbb{E}[|H_k^c(m,n)|^2]$ is
assumed to be constant w.r.t~$n$). It is thus reasonable to assume that the base
station modulates the subcarriers of each user in each one of the two subsets
($\cal I$ and $\mathcal{P}_c$) with the same transmit power. Define $P\ku^c$
(resp. $P\kd^c$) as the power transmitted on the subcarriers assigned to
user~$k$ in $\mathcal{I}$ (resp. in $\mathcal{P}_c$) \emph{i.e.,} $P\ku^c =
E[|s_k(n,m)|^2]$ if $n\in \mathcal{I}$, $P\kd^c = E[|s_k(n,m)|^2]$ if $n\in
\mathcal{P}_c$. Parameters $\{\gamma\ki^c,P\ki^c\}_{i=1,2}$ will be designated
in the sequel as the \emph{resource allocation parameters}. We now describe the
adopted model for the multicell interference. Consider one of the non protected
subcarriers~$n$ assigned to user~$k$ of cell~$A$ in subset~$\mathcal{I}$. Denote
by~$\sigma_k^2$ the variance of the additive noise process~$w_k(n,m)$. This
variance is assumed to be constant w.r.t both~$n$ and~$m$. It depends only on
the position of user~$k$ and the average powers
$Q_1^B=\sum_{k=1}^{K^B}\gamma\ku^B P\ku^B$ and
$Q_1^C=\sum_{k=1}^{K^C}\gamma\ku^C P\ku^C$ transmitted respectively by base
stations~$B$ and~$C$ on the subcarriers of~$\cal I$. This assumption is valid
for instance in OFDMA systems that utilize \emph{random subcarrier
assignment}~\cite{random_SAS}. According to this subcarrier assignment scheme,
each user~$k$ is assigned a subset~$\mathcal{N}_k$ that is composed by
\emph{randomly} selecting $\card{\mathcal{N}_k}$ subcarriers out of the total
$N$ available subcarriers. Finally, let $\sigma^2$ designate the variance of the
thermal noise. Putting all pieces together:
\begin{equation}\label{eq:2d_multicell_interference}
\mathbb{E}\left[|w_k(n,m)|^2\right]=
\left\{
\begin{array}{ll}
\sigma^2 &\mbox{if } n\in \mathcal{P}_c\\
\sigma_k^2=\sigma^2+\sum_{\tilde{c}=B,C}\mathbb{E}\left[
|H_k^{\tilde{c}}(n,m)|^2\right] Q_1^{\tilde{c}} &\mbox{if } n\in\mathcal{I}
\end{array}
\right.
\end{equation}
where $H_k^{\tilde{c}}(n,m)$ ($\tilde{c}=B,C$) represents the channel between
base station~$\tilde{c}$ and user~$k$ in cell~$c$ on subcarrier~$n$ and OFDM
block~$m$. Of course, the average channel gain
$\mathbb{E}\left[|H_k^{\tilde{c}}(n,m)|^2\right]$ depends on the position of
user~$k$ via the path loss model. For instance, if two users~$k$ and~$l$ of
cell~$A$ are located on the same line perpendicular to the axis $BC$ such
that~$k$ is closer to base station~$A$, then $\sigma_k^2\leq \sigma_l^2$.

\section{Joint Resource Allocation Problem}
\label{sec:joint_2d}

Assume that each user~$k$ has a rate requirement of $R_k$ nats/s/Hz. Consider
the problem of determination of the resource allocation parameters for the three
interfering sectors. These parameters must be selected such that the target rate
of each user is satisfied and such that the power spent by the three base
stations is minimized. Due to the ergodicity of the process $\{H_k^c(n,m)\}_m$
for each subcarrier~$n$, the rate $R_k$ can be satisfied provided that it is
smaller than the ergodic capacity $C_k$ associated with user~$k$. Unfortunately,
the exact expression of $C_k$ is difficult to obtain due to the fact that the
noise-plus-interference $\{w_k(n,m)\}_{n,m}$ is not a Gaussian process
in general. Nonetheless, if we endow the input symbols $s_k(n,m)$ with Gaussian
distribution, the mutual information between $s_k(n,m)$ and the received signal
$y_k(n,m)$ in~\eqref{eq:signal_model} is minimal when $w_k(n,m)$ is Gaussian
distributed. Therefore, we approximate in the sequel the multicell interference
by a Gaussian process as this approximation provides a lower bound on the
mutual information. Focus on cell~$A$ and denote by $g_{k,1}(Q_1^B,Q_1^C)$,
$g_{k,2}$ the channel Gain-to-Interference-plus-Noise Ratio (GINR) and
Gain-to-Noise Ratio (GNR) associated with user~$k$ on the subcarriers of subset
$\mathcal{I}$ and $\mathcal{P}_A$ respectively:
$$
g\ku(Q_1^B,Q_1^C) = \frac{\rho_k}{\sigma_k^2}\:, \qquad
g\kd = \frac{\rho_k}{\sigma^2}\:.
$$
The ergodic capacity $C_k$ associated with user~$k$ in cell~$A$ is equal to the
sum of the ergodic capacities corresponding to both subsets~$\cal I$ and~${\cal
P}_A$. For instance, the part of the capacity corresponding to the protected
subset~${\cal P}_A$ is equal to
$\gamma\kd^A\EE{\log\left(1+P\kd^A\frac{|H_k^{A}(n,m)|^2}{\sigma^2}\right)}$,
where factor $\gamma\kd^A$ traduces the fact that the capacity increases with
the number of subcarriers which are modulated by user~$k$. In the latter
expression, the expectation is calculated w.r.t random variable
$\frac{|H_k^{A}(m,n)|^2}{\sigma^2}$. Now, $\frac{H_k^{A}(m,n)|^2}{\sigma^2}$ has
the same distribution as $\frac{\rho_k}{\sigma^2}Z=g\kd Z$, where $Z$ follows a
standard unit-variance exponential distribution. Finally, the ergodic capacity
$C_k=C_k(\gamma\ku^A,\gamma\kd^A,P\ku^A,P\kd^A,Q_1^B,Q_1^C)$ in the whole
bandwidth is equal to
\begin{equation}\label{eq:2d_ergodic_capacity}
C_k=\gamma\ku^A\mathbb{E}\left[\log\left(1+g\ku(Q_1^B,Q_1^C)
P\ku^A Z\right)\right]
+\gamma\kd^A\mathbb{E}\left[\log\left(1+g\kd P\kd^A Z\right)\right]\:.
\end{equation}
Capacity~$C_k$ is achieved if we endow the input symbols $s_{k}(n,m)$
with Gaussian distribution. This distribution is assumed from now on. Moreover,
note that $C_k$ does not depend on the particular subcarriers $\mathcal{N}_k$
assigned to user~$k$,  but rather on the number of these subcarriers via
parameters $\gamma\ku^A$ and $\gamma\kd^A$. Therefore, choosing some specific
subcarriers rather than others has no effect on the capacity. The subcarriers
assignment scheme reduces thus to the determination of the sharing factors
$\gamma\ku^A$, $\gamma\kd^A$. Finally, the multicell resource allocation problem
can be defined as follows.
\begin{problem}
\label{prob:2d_multi}
Minimize the power spent by the three base stations 
$Q=$ $\dsp \sum_{c=A,B,C}\sum_{k=1}^{K^c}(\gamma\ku^c P\ku^c+
\gamma\kd^c P\kd^c)$ w.r.t
$\{\gamma\ku^c, \gamma\kd^c, P\ku^c, P\kd^c\}_{\substack{c=A,B,C\\ k=1\ldots
K^c}}$ under the following constraints:
\begin{align*}
\mathbf{C1:}\: &\forall k, C_k(\gamma\ku^c,\gamma\kd^c,P\ku^c,P\kd^c)\geq R_k &
&\mathbf{C3:}\:
\sum_{k=1}^{K^c}\gamma\kd^c=\frac{1-\alpha}{3}\\
\mathbf{C2:}\: &\sum_{k=1}^{K^c}\gamma\ku^c=\alpha & &\mathbf{C4:}\:
\forall k,\: \gamma\ki^c,P\ki^c\geq0\: (i=1,2)\:.
\end{align*}
\end{problem}
As a matter of fact, Problem~\ref{prob:2d_multi} cannot be solved using convex
optimization tools. Anyhow, even if we were able to propose a method to solve
this problem (as we did in~\cite{tsp1} in the case of 1-D networks), such a
method would be very costly in term of computational complexity. It is therefore
of interest to propose practical allocation algorithms that provides suboptimal
solutions to Problem~\ref{prob:2d_multi}.


\section{Proposed Resource Allocation Algorithm}
\label{sec:sub_optimal_2d}

In~\cite{tsp2}, we showed that in 1-D cellular networks, any global
solution to Problem~\ref{prob:2d_multi} has the following \emph{asymptotic}
property: The power allocated to users who modulate \emph{both} protected and
non protected subcarriers becomes negligible as the number $K$ of users
increases. One can thus suggest the suboptimal (w.r.t
Problem~\ref{prob:2d_multi}) resource allocation algorithm given below.
For a given user $k$ in cell~$c$, we denote by $(x_k,y_k)$ his/her position in
the Cartesian coordinate system associated with this cell (see
Figure~\ref{fig:subopt_2d}). In our algorithm, we use a continuous
function $d_{\textrm{subopt}}^c(.)$ on $[-D,D]$ (where $D$ stands for the radius
of the cell as shown in Figure~\ref{fig:subopt_2d}) to define a curve that
separates the users of each cell~$c$ into two subsets. The first subset
$\mathcal{K}_I^c$ contains the users who are closer to the base station than
this curve. These users are constrained to modulate only non protected
subcarriers~$\cal I$. The second subset $\mathcal{K}_P^c$ contains the rest of
users who are constrained to the protected subcarriers~$\mathcal{P}_c$:
$$
\mathcal{K}_I^c=\{k\in\{1 \ldots K^c\}\:|\: y_k\leq
d_{\textrm{subopt}}^c(x_k)\}\:,\qquad
\mathcal{K}_P^c=\{k\in\{1 \ldots K^c\}\:|\: y_k>
d_{\textrm{subopt}}^c(x_k)\}\:.
$$
Note that $\{d_{\textrm{subopt}}^c(.)\}_{c=A,B,C}$ are fixed prior to
resource allocation. Relevant selection of these curves is postponed
to Subsection~\ref{sec:selection_d}. It merely relies on the asymptotic
analysis carried out in Subsection~\ref{sec:optimal_asyp_analysis_2d_aligned}.
\subsection{Resource Allocation for Interfering Users
$\boldsymbol{\{{\cal K}^c_I\}_{c=A,B,C}}$}
\label{sec:alloc_for_inter_2d}

For users ${\cal K}^c_I$ in each cell~$c$, resource allocation parameters in the
protected subset $\mathcal{P}_c$ are arbitrarily set to zero \emph{i.e.,}
$\gamma\kd^c=P\kd^c=0$. Recall the definition of 
$Q_1^c = \sum_{k\in{\cal K}_I^c}\gamma\ku^cP\ku^c$ as the average power
transmitted by base station~$c$ ($c=A,B,C$) in the interference subset $\cal I$.
For each cell~$c$, denote by $\bar{c}$ and $\bar{\bar{c}}$ the other two cells.
For example, $\bar{A}=B$ and $\bar{\bar{A}}=C$. Define $C_k($ $\gamma\ku^c$,
$P\ku^c$, $Q_1^{\bar{c}}$, $Q_1^{\bar{\bar{c}}})$ as the ergodic capacity
associated with user~$k$ obtained by plugging $\gamma\kd^c=P\kd^c=0$
into~\eqref{eq:2d_ergodic_capacity}. Parameters $\gamma\ku^c, P\ku^c$ for users
in $\{{\cal K}_I^c\}_{c=A,B,C}$ can be obtained as the solution to the following
multicell allocation problem.
\begin{problem}
\label{prob:opt_multi_2d}
{\bf [Multicell problem in band $\boldsymbol{\mathcal{I}}$]} 
Minimize the total transmit power $\dsp \sum_{c=A,B,C} \sum_{k\in {\cal K}_I^c}
\gamma\ku^c P\ku^c$ w.r.t. $\{\gamma\ku^c$, $P\ku^c\}_{\stackrel{c=A,B,C}{k= 1
\ldots K^c}}$ under the following constraints:
\begin{align*}
&\mathbf{C1:}\: \forall c, \; \forall k\in{\cal K}_I^c, 
R_k\leq C_k(\gamma\ku^c,P\ku^c,Q_1^{\bar{c}},Q_1^{\bar{\bar{c}}})\\
&\mathbf{C2:}\: \forall c, \; \sum_{k\in{\cal K}_I^c} \gamma\ku^c= \alpha\qquad
\mathbf{C3:}\: \forall c, \forall k\in{\cal K}_I^c,\: \gamma\ku^c,
P\ku^c\geq0\:.
\end{align*}
\end{problem}
\begin{remark}
\label{rem:feasibility}
Problem~\ref{prob:opt_multi_2d} may not be always feasible. Indeed, since the
protected subcarriers are forbidden to users $\mathcal{K}_I^c$, the multicell
interference may in some cases reach excessive levels and prevent some users
from satisfying their rate requirements. Fortunately, we will see that if curves
$\{d_{\textrm{subopt}}^c(.)\}_{c=A,B,C}$ are relevantly chosen, then the latter
problem is feasible, at least for a sufficiently large number of users. 
\end{remark}
One can use an approach similar to~\cite{tsp1,tsp2} to show that any global
solution to the above problem satisfies the following property.
There exist six positive numbers $\{\beta_1^c, Q_1^c\}_{c=A,B,C}$ (where
$\beta_1^c$ is the Lagrange multiplier associated with constraint $\mathbf{C2}$
of Problem~\ref{prob:opt_multi_2d}) such that:
\begin{align}
& P\ku^c = \left[g\ku(Q_1^{\bar{c}}, Q_1^{\bar{\bar{c}}})\right]^{-1}
f^{-1}(g\ku(Q_1^{\bar{c}}, Q_1^{\bar{\bar{c}}}) \beta_1^c)
\label{eq:resSimpa_2d}\\
& \gamma\ku^c = \frac{R_k}{C\left(
g\ku(Q_1^{\bar{c}}, Q_1^{\bar{\bar{c}}})
\beta_1^c\right)}\:,\label{eq:resSimpb_2d}
\end{align}
where $f(.)$ and $C(.)$ are increasing functions defined on ${\mathbb R}_+$ by  
\begin{equation}
f(x) = \frac{\EE{\log(1+xZ)}}{\EE{\frac{Z}{1+xZ}}}-x\:,\qquad
C(x)=\mathbb{E}[\log(1+f^{-1}(x)Z)]\:,\label{eq:fC_fun}
\end{equation}
$f^{-1}(.)$ being the inverse on $\mathbb{R}_+$ of $f(.)$ w.r.t the composition
of functions, and where for each $c=A,B,C$ and for a fixed value of
$Q_1^{\bar{c}}$ and $Q_1^{\bar{\bar{c}}}$, $(\beta_1^c, Q_1^c)$ is the unique
solution to the following system of equation:
\begin{align}
& \sum_{k\in{\cal K}_I^c} \frac{R_k}{
C(g\ku(Q_1^{\bar c}, Q_1^{\bar{\bar c}})
\beta_1^c)} = \alpha\:,\label{eq:simpa_2d}\\
& Q_1^c = \sum_{k\in{\cal K}_I^c} R_k
\frac{\left[g\ku(Q_1^{\bar c}, Q_1^{\bar{\bar c}})\right]^{-1}
f^{-1}(g\ku(Q_1^{\bar c}, Q_1^{\bar{\bar c}}) \beta_1^c)}
{C(g\ku(Q_1^{\bar{c}}, Q_1^{{\bar{\bar c}}})\beta_1^c)}\:. 
\label{eq:simpb_2d}
\end{align}
Note that equation~\eqref{eq:simpa_2d} is equivalent to the constraint 
$\bf C2$: $\sum_k \gamma\ku^c$ $=$ $\alpha$, while equation~\eqref{eq:simpb_2d}
is nothing else than the definition of the average power $Q_1^c =
\sum_{k\in{\cal K}_I^c} \gamma\ku^cP\ku^c$ transmitted by base station~$c$ in
subset~$\mathcal{I}$. We now prove that when Problem~\ref{prob:opt_multi_2d} is
feasible, then the system of six equations
(\ref{eq:simpa_2d})-(\ref{eq:simpb_2d}) for $c=A,B,C$ admits a unique solution
$\beta_1^A$, $Q_1^A$, $\beta_1^B$, $Q_1^B,\beta_1^C$, $Q_1^C$ and that this
solution can be
obtained by a simple iterative algorithm. Focus on a given cell~$c$ ($c=A,B,C$)
and consider any fixed values  $Q_1^{\bar c}$, $Q_1^{\bar{\bar c}}$. Denote by 
$I^c\left(Q_1^{\bar c},Q_1^{\bar{\bar c}}\right)$ the rhs of~\eqref{eq:simpb_2d}
\emph{i.e.,}
$$
I^c\left(Q_1^{\bar c},Q_1^{\bar{\bar c}}\right) = \sum_{k\in{\cal K}_I^c} R_k
\frac{\left[g\ku(Q_1^{\bar c}, Q_1^{\bar{\bar c}})\right]^{-1}
f^{-1}(g\ku(Q_1^{\bar c}, Q_1^{\bar{\bar c}}) \beta_1^c)}
{C(g\ku(Q_1^{\bar{c}}, Q_1^{{\bar{\bar c}}})\beta_1^c)}\:,
$$
where $\beta_1^c$ is defined as the unique solution to~\eqref{eq:simpa_2d}. The
value $I^c\left(Q_1^{\bar c},Q_1^{\bar{\bar c}}\right)$ can be seen as the
minimum power that should be spent by base station~$c$ on the interference
subcarriers~$\cal I$ when the interference produced by base stations~${\bar c}$
and~$\bar{\bar c}$ is equal to $Q_1^{\bar c}$ and $Q_1^{\bar{\bar c}}$,
respectively. Since~(8) should be satisfied for $c=A$, $c=B$ and~$c=C$, the
following three equations hold
\begin{equation*}
Q_1^A=I^A(Q_1^B,Q_1^C),\quad Q_1^B=I^B(Q_1^A,Q_1^C),\quad
Q_1^C=I^C(Q_1^A,Q_1^B)\:.
\end{equation*}
The triple $(Q_1^A, Q_1^B,Q_1^C)$ is therefore clearly a fixed point of the
vector-valued function ${\bf I}(Q_1^A, Q_1^B,Q_1^C)$ $=$
$\left(I^A(Q_1^B,Q_1^C), I^B(Q_1^A,Q_1^C),I^C(Q_1^A,Q_1^B)\right)$:
\begin{equation}\label{eq:fixedtilde}
(Q_1^A, Q_1^B,Q_1^C) = {\bf I}(Q_1^A, Q_1^B, Q_1^C)\:.   
\end{equation}
As a matter of fact, it can be shown that such a fixed point of ${\bf I}$ is
unique. This claim can be proved using the following lemma. 
\begin{lemma}
\label{lem:yates}  
Function ${\bf I}$ is such that the following properties hold.
  \begin{enumerate}
  \item Positivity: ${\mathbf{I}}(Q^A,Q^B,Q^C)>0$.
  \item Monotonicity: If $Q^A\geq {Q^A}', Q^B\geq{Q^B}', Q^C\geq{Q^C}'$, then 
  ${\mathbf{I}}(Q^A,Q^B,Q^C)\geq {\mathbf{I}}({Q^A}',{Q^B}',{Q^C}')$.
  \item Scalability: for all $t>1$,
  $t {\mathbf{I}}(Q^A,Q^B,Q^C)>{\mathbf{I}}(tQ^A,tQ^B,tQ^C)$.
  \end{enumerate}
\end{lemma}
The proof of Lemma~\ref{lem:yates} uses arguments which are very similar to
the proof of Theorem~1 in~\cite{papandriopoulos}. Function ${\mathbf{I}}$
is then a {\it standard interference function}, using the terminology
of~\cite{yates}. Therefore, as stated in~\cite{yates}, such a function
${\bf I}$ admits at most one fixed point. On the other hand, the existence 
of a fixed point is ensured by the feasibility of
Problem~\ref{prob:opt_multi_2d} and by the fact
that~(\ref{eq:fixedtilde}) holds for any global solution. In other words, if
Problem~\ref{prob:opt_multi_2d} is feasible, then function ${\bf I}$ does admit
a fixed point and this fixed point is unique. In the latter case, the results
of~\cite{yates} state furthermore that a simple fixed point algorithm (such as
Algorithm~\ref{algo:ping_pong_2d} given below) applied to function $\mathbf{I}$
converges necessarily to its unique fixed point.
\begin{remark}
Note that in Algorithm~\ref{algo:ping_pong_2d}, the only information needed by
each base station~$c$ ($c=A,B,C$) about the other two cells $\bar{c}$,
$\bar{\bar{c}}$ is the current value of the powers $Q_1^{\bar{c}}$,
$Q_1^{\bar{{\bar c}}}$ transmitted in the interference band $\cal I$. This value
can \emph{i)} either be measured by base station~$c$ at each iteration of
Algorithm~\ref{algo:ping_pong_2d}, or \emph{ii)} it can be communicated to it by
base stations~$\bar{c}$ and $\bar{\bar{c}}$ over a dedicated link. In the first
case, no message passing is required, and in the second case only few
information is exchanged between the base stations.
Algorithm~\ref{algo:ping_pong_2d} can thus be implemented in a distributed
fashion.
\end{remark}
Of course, the feasibility of Problem~\ref{prob:opt_multi_2d} depends on the
choice of the separating curves $\{d_{\textrm{subopt}}^c(.)\}_{c=A,B,C}$.
Section~\ref{sec:asym_analysis_2d} addresses the relevant selection of these
curves such that Algorithm~\ref{algo:ping_pong_2d} converges for a sufficiently
large number of users. 
\subsection{Resource Allocation for Protected Users 
$\boldsymbol{\{{\cal K}^c_P\}_{c=A,B,C}}$}
\label{sec:alloc_for_noninter_2d}

Since users ${\cal K}^c_P$ in each cell~$c$ are constrained to modulate only the
subcarriers of subset $\mathcal{P}_c$, they are not subject to
multicell interference. Resource allocation for such users can thus be done
independently in each cell by solving a simple single cell optimization problem
which is a special case of Problem~\ref{prob:opt_multi_2d}. Focus for example
on cell~$A$.
One can show~\cite{tsp1} that the resource allocation problem for users of this
cell is convex in variables $\{\gamma\kd^A, w\kd^A\}_{k\in {\cal K}^A_P}$, where
$w\kd^A=\gamma\kd^A P\kd^A$. Its solution can be obtained by solving
the associated KKT conditions and is given by:
\begin{align}
& P\kd^A=g\kd^{-1}f^{-1}(g\kd \beta_2^A) \label{eq:power_protected_2d}\\
& \gamma\kd^A=\frac{R_k}{C\left(g\kd \beta_2^A\right)}\:.
\label{eq:gamma_protected_2d}
\end{align}
Parameter $\beta_2^A$ is obtained by writing that constraint
$\sum_k \gamma\kd^A = \frac{1-\alpha}{3}$ holds as the unique solution to:
\begin{equation}\label{eq:beta_2_tilde_2d}
\sum_{k\in {\cal K}^A_P} \frac{R_k}{C(g\kd \beta_2^A)} = \frac{1-\alpha}{3}\:.
\end{equation}
Resource allocation parameters for users of cells~$B$ and~$C$ can be
similarly obtained. The following procedure performs the above resource
allocation for protected users.
\subsection{Summary: Distributed Resource Allocation Algorithm}

The proposed distributed resource allocation scheme is finally summarized
by Algorithm~\ref{algo:suboptimal_2d}.
\subsection{Complexity Analysis}
\label{sec:complexity}

By referring to Algorithm~\ref{algo:protected_2d}, it is straightforward to
verify that resource allocation for protected users can be reduced to the
determination in each cell~$c$ of the value of $\beta_2^c$, which is the unique
solution to the equation 
$\sum_{k\in {\cal K}^c_P} \frac{R_k}{C(g\kd \beta_2^c)} = \frac{1-\alpha}{3}$.
Since function~$x\mapsto 1/C(x)$ is convex, the latter solution can be
numerically obtained by any of the classical zero-finding algorithms of the
convex optimization literature such as the gradient method~\cite{grad}. Denote
by~$N_{\textrm{grad}}$ the number of iterations required till the convergence
of such a method. Each one of these iterations requires a computational
complexity proportional to the number of terms in the lhs of the equation. The
overall computational complexity of finding $\beta_2^c$ is therefore of order
$O(N_{\textrm{grad}} K)$. In the same way, one can show that each iteration of
Algorithm~\ref{algo:ping_pong_2d} can be performed with a complexity of order
$O(N_{\textrm{grad}} K)$. Let $N_{\textrm{iter}}$ designate the number of
iterations of Algorithm~\ref{algo:ping_pong_2d} needed till convergence (within
a certain accuracy). The overall computational complexity of
Algorithm~\ref{algo:ping_pong_2d}, and hence of
Algorithm~\ref{algo:suboptimal_2d} as well, is thus of the order of
$O(N_{\textrm{iter}} N_{\textrm{grad}} K)$. Our simulations showed that
Algorithm~\ref{algo:ping_pong_2d} converges relatively quickly
in most of the cases. Indeed, no more than $N_{\textrm{iter}}=15$ iterations
were needed to reach convergence within a very reasonable accuracy in most of
the practical situations.
\section{Determination of Curves
$\boldsymbol{\{\MakeLowercase{d}_{\textrm{\MakeLowercase{subopt}}}^{
\MakeLowercase{c}}(.)\}}$ and Asymptotic Optimality
of Algorithm~3}
\label{sec:asym_analysis_2d}

The aim of this section is to relevantly select the separating curves
$d_{\textrm{subopt}}^A(.)$, $d_{\textrm{subopt}}^B(.)$ and
$d_{\textrm{subopt}}^C(.)$. For that sake, we consider the case where the
number~$K$ of users tends to infinity in a sense that will be clear later on,
and we prove Theorem~\ref{the:subopt_2d} (see Subsection~\ref{sec:asym_opt_2d})
which states the following.
\framebox{{\parbox{\linewidth}{
\centering
There exist curves $\{d_{\textrm{subopt}}^c(.)\}_{c=A,B,C}$ such that the
transmit power of Algorithm~\ref{algo:suboptimal_2d} converges as $K\to\infty$
to the limit total power of an optimal solution to the joint allocation problem
(Problem~\ref{prob:2d_multi}).
}}}
Otherwise stated, Algorithm~\ref{algo:suboptimal_2d} is asymptotically optimal
if the separating curves are well chosen. In order to prove this result, we
first characterize the form and the total transmit power $Q_T\K$ of an {\bf
optimal} solution to Problem~\ref{prob:2d_multi} in the special case where users
of each cell are aligned on parallel equispaced lines. Indeed, we prove that the
latter solution has the following ``binary'' property: In each cell~$c$, there
exists a curve~$d_{\theta\cK}$ that separates users modulating uniquely
protected or non protected subcarriers. Here, $\theta^{c,(K)}$ is a vector of
parameters that will be specified later on and which depends on the system
setting (including the number $K$ of users). We show that as the number~$K$ of
users tends to infinity, $d_{\theta\cK}$ converges, at least for certain
subsequences~$(K)$, to a curve $d_{\theta^c}$ that can be characterized by
solving a certain system of equations. The same system allows to compute the
limit $Q_T=\lim_{K\to\infty}Q_T\K$. Next, we consider the case of an arbitrary
geographical distribution where users are not necessarily aligned on parallel
lines. Eventhough the aforementioned binary property no longer holds in this
general case, we show that the transmit power of an optimal solution to
Problem~\ref{prob:2d_multi} converges to the same limit $Q_T$ as in the case of
aligned users. This result will suggest to relevantly select the separating
curves $d_{\textrm{subopt}}^c(.)$ of the \emph{suboptimal} allocation algorithm
to be equal to the asymptotic \emph{optimal} curves $d_{\theta^c}(.)$. Thanks to
the latter curve selection, we prove that the proposed allocation algorithm
becomes asymptotically optimal.
 
\subsection{Asymptotic Optimal Allocation}
\label{sec:optimal_asyp_analysis_2d_aligned}

The characterization of the asymptotic behaviour of an optimal solution to the
joint resource allocation problem is performed by the following three steps. 

\subsubsection{Step 1: Single Cell Resource Allocation}

We first consider a particular case where users of each cell are aligned on
equispaced parallel lines. Focus for example on cell~$A$ and define $I^A$
parallel equispaced lines ($I^A < K^A$) which pass through cell~$A$ and which
are perpendicular to the axis~$BC$ as illustrated in
Figure~\ref{fig:finite_separating_curve}. Next, assign each one of these lines
an index $i\in\{1 \ldots I^A\}$. In the sequel, we denote by
$\mathcal{L}_i^A\subset\{1 \ldots K^A\}$ the subset composed of the users of
cell~$A$ located on the line whose index is~$i$. Assume that the resource
allocation parameters of users of cells~$B$ and~$C$ are fixed and recall the
definition of $C_k$ given by~\eqref{eq:2d_ergodic_capacity} as the
ergodic capacity of user~$k$. The optimal resource allocation problem for
cell~$A$ consists in characterizing $\{\gamma\ku^A$, $\gamma\kd^A$, $P\ku^A$,
$P\kd^A\}_{k=1\ldots K^A}$ allowing to satisfy the rate requirements of all
users $k\in\{1, \ldots, K^A\}$. The determination of these parameters should be
done such that the power 
$Q^A=\sum_{k=1}^{K^A}\gamma\ku^A P\ku^A+\gamma\kd^A P\kd^A$ to be spent is
minimum:
\begin{problem}
\label{prob:single_2d}
Minimize $Q^A=\sum_{k=1}^{K^A}\gamma\ku^A P\ku^A+\gamma\kd^A P\kd^A$ with
respect to 
$\{\gamma\ku^A$, $\gamma\kd^A$, $P\ku^A$, $P\kd^A\}_{k = 1\ldots K^A}$ 
under the following constraints:
\begin{align*}
\mathbf{C1:}\: &\forall k,R_k\leq C_k & 
&\mathbf{C4:}\: \gamma\ku^A\geq0,\gamma\kd^A\geq0\\
\mathbf{C2:}\: &\sum_{k=1}^{K^A}\gamma\ku^A = \alpha & 
&\mathbf{C5:}\: P\ku^A\geq0,P\kd^A\geq0.\\
\mathbf{C3:}\: &\sum_{k=1}^{K^A}\gamma\kd^A = \frac{1-\alpha}{3} & 
& {\bf C6:} \: \sum_{k=1}^{K^A} \gamma\ku^A P\ku^A \leq {\cal Q}\:.
\end{align*}
\end{problem}
Here, constraint~$\bf C6$ is a ``low nuisance constraint'' which is introduced
to limit the interference \emph{produced} by Base Station~$A$. In other words,
the power $Q_1^A=\sum_k \gamma\ku^A P\ku^A$ which is transmitted by base
station~$A$ on the subcarriers of subset~$\cal I$ should not exceed a certain
\emph{nuisance level}~$\cal Q$. The introduction of~$\bf C6$ is a technical
tool revealed to be useful in solving the multicell allocation problem later
on. On one hand, note that Problem~\ref{prob:single_2d} is feasible for any
$\alpha>0$ and ${\cal Q}\geq 0$ since it has at least the following trivial
solution. The solution consists in assigning zero power $P\ku^A=0$ on the
subcarriers of subset~$\cal I$ (so that constraint $\bf C6$ will be satisfied),
and in performing resource allocation only using the subcarriers of
subset~$\mathcal{P}_A$. On the other hand, Problem~\ref{prob:single_2d} can be
made convex after a slight change of variables, as a matter of fact. Therefore,
any global solution to this problem is characterized by the KKT conditions. The
simplification of these conditions is not presented in this paper due to lack of
space. However, it can be done in a very similar way as in the case of 1-D
cellular networks addressed in our previous work~\cite{tsp1} leading to the
following result. Resource allocation parameters of any of the subsets
$\mathcal{L}_i^A$ of users located on lines $i=1\ldots I^A$ have a ``binary''
separation property as the users of a 1-D cell. This property is summarized
below. Define the following decreasing function for each $x\in\mathbb{R}_+$: 
\begin{equation}\label{eq:F_fun}
F(x) = \EE{\frac{Z}{1+f^{-1}(x)Z}}\:,
\end{equation}
and let $\beta_1$, $\beta_2$ and $\xi$ designate the Lagrange multipliers
associated with constraints $\bf C2$, $\bf C3$ and $\bf C6$ respectively.
There exists a ``pivot-position'' on each line~$i$ such that users
$k\in\mathcal{L}_i^A$ who are farther than this position are
uniquely assigned subcarriers from the protected subset $\mathcal{P}_A$ (by
setting $\gamma\ku^A=0$). Moreover, such ``protected users'' satisfy:
\begin{equation}\label{eq:single_above}
\frac{g\ku\left(Q_1^B,Q_1^C\right)}{1+\xi}
F\left(\frac{g\ku\left(Q_1^B,Q_1^C\right)}{1+\xi}\beta_1\right)< g\kd F(g\kd
\beta_2)\:.
\end{equation}
On the other hand, users $k\in\mathcal{L}_i^A$ who are closer to the base
station than the pivot-position are uniquely assigned interference subcarriers
from subset $\mathcal{I}$ (by setting $\gamma\kd^A=0$). Such ``non
protected users'' satisfy:
\begin{equation}\label{eq:single_under}
\frac{g\ku\left(Q_1^B,Q_1^C\right)}{1+\xi}
F\left(\frac{g\ku\left(Q_1^B,Q_1^C\right)}{1+\xi}\beta_1\right)> g\kd F(g\kd
\beta_2)\:.
\end{equation} 
The proof of the above separation property uses Conjecture~1 in~\cite{tsp1}
which can be easily validated numerically. Inequalities~\eqref{eq:single_above}
and~\eqref{eq:single_under} suggest the definition of a curve that
geographically separates protected from non protected users of cell~$A$. This
can be done as follows. We write the variance $\rho_k$ of the channel gain of
user~$k$ as $\rho_k=\rho(x_k,y_k)$ where $\rho(x,y)$ models the path loss.
Function $\rho(x,y)$ is assumed to have the form $\rho(x,y)=\eta
(\sqrt{x^2+y^2})^{-s}$ where $\sqrt{x^2+y^2}$ is the distance separating $(x,y)$
from the base station, $\eta$ is a certain gain and $s$ is the path-loss
coefficient. We also denote by $g_2(x,y) = \frac{\rho(x,y)}{\sigma^2}$ the GNR
on the protected subcarriers associated with a user at position $(x,y)$. Note
that for any user~$k$, $g_2(x_k,y_k) = g\kd$. In the same way, $g_1(x,y,{\cal
Q}',{\cal Q}'')$ denotes the GINR at position $(x,y)$ if the interfering base
stations are transmitting with power ${\cal Q}'$ and ${\cal Q}''$ on the
interference subcarriers $\cal I$. Using the above notation, we have
$g_1\left(x_k,y_k,Q_1^B,Q_1^C\right) = g\ku\left(Q_1^B,Q_1^C\right)$ for each
user~$k$ in cell~$A$. Note that for any $(x,y)$, $g_2(x,y)=g_1(x,y,0,0)$.
Finally, for each 
$\theta=(\beta_1,\beta_2,\mathcal{Q}',\mathcal{Q}'',\xi)\in\mathbb{R}_+^5$, 
we define
\begin{equation}\label{eq:H_fun}
W_{\theta}(x,y)=
\frac{g_1(x,y,\mathcal{Q}',\mathcal{Q}'')}{1+\xi} 
F\left(\frac{g_1(x,y,\mathcal{Q}',\mathcal{Q}'')}{1+\xi}
\beta_1\right)-g_2(x,y) F(g_2(x,y) \beta_2)\:.
\end{equation}
Due to~\eqref{eq:single_above}, we have $W_{\theta}(x_k,y_k)<0$ for each
protected user~$k$ \emph{i.e.,} for users farther from the base
station than the pivot-position. Inversely, $W_{\theta}(x_k,y_k)>0$ for each non
protected user~$k$ \emph{i.e.,} for users closer to the base station than the
pivot-position. Therefore, function $d_{\theta}(x)$ given below defines the
curve that we are seeking and which geographically separates protected from non
protected users of cell~$A$:
\begin{equation}\label{eq:separating_curve_fini}
\begin{multlined}
d_{\theta}(x)=\left\{
\begin{array}{ll}
\frac{|x|}{\sqrt{3}} &\textrm{if }
W_{\theta}\left(x,\frac{|x|}{\sqrt{3}}\right)<0\\
\frac{2D-|x|}{\sqrt{3}} &\textrm{if }
\min\left\{W_{\theta}\left(x,\frac{|x|}{\sqrt{3}}\right),
W_{\theta}\left(x,\frac{2D-|x|}{\sqrt{3}}\right)\right\}>0\\
\textrm{the unique zero of } y\mapsto W_{\theta}(x,y)
&\textrm{otherwise}\:.
\end{array}
\right.
\end{multlined}
\end{equation}
Note in particular that the first two conditions
of~\eqref{eq:separating_curve_fini} hold in the case where the pivot-position at
line $x$ is located at the upper sector border $y=|x|/\sqrt{3}$ or the lower
sector border $y=\left(2D-|x|\right)/\sqrt{3}$. When these two conditions are
\emph{not} satisfied, the existence of the zero of the continuous function 
$y\mapsto W_{\theta}(x,y)$ is straightforward due to the intermediate value
theorem. The uniqueness of this zero can be proved by arguments already
developed in the proof of Lemma~1 in~\cite{tsp1}. Finally, we obtain the
following lemma. 
\begin{lemma}
\label{lem:single_2d} 
Assume that the users of cell~$A$ are aligned on $I^A$ parallel equispaced
lines (as in Figure~\ref{fig:finite_separating_curve}) and that the power
transmitted by base stations~$B$ and~$C$ on the non protected subcarriers $\cal
I$ is set to $Q_1^B$ and~$Q_1^C$ respectively. The global solution 
$\{\gamma\ku^A,\gamma\kd^A,P\ku^A,P\kd^A\}_{k=1\ldots K^A}$ to
Problem~\ref{prob:single_2d} is unique and is as follows. There exist three
unique nonnegative numbers $\beta_1$, $\beta_2$, $\xi$ such that:
\begin{enumerate}
\item For each $k\in{{\cal L}_i^A}$ such that $y_k<d_{\theta}(x_k)$,
\begin{equation}
\begin{array}[h]{l|l}
\dsp
P\ku^A=\left[g\ku\left(Q_1^B,Q_1^C\right)\right]^{-1}f^{-1}
\left(\frac{g\ku\left(Q_1^B,Q_1^C\right)}{1+\xi} \beta_1\right) &
P\kd^A=0 \\
\dsp
\gamma\ku^A=\frac{R_k}{C\left(\frac{g\ku\left(Q_1^B,Q_1^C\right)}{1+\xi}
\beta_1\right)} &
\gamma\kd^A=0
\end{array}\label{eq:allocinf_2d}
\end{equation}
\item For each $k\in{{\cal L}_i^A}$ such that $y_k>d_{\theta}(x_k)$,
\begin{equation}
\label{eq:allocsup_2d}
\begin{array}[h]{l|l}
P\ku^A=0 & \dsp P\kd^A=g\kd^{-1}f^{-1}(g\kd \beta_2) \\
\gamma\ku^A=0 & \dsp 
\gamma\kd^A=\frac{R_k}{C\left(g\kd \beta_2\right)}
\end{array}
\end{equation}
\end{enumerate}
where $\beta_1$, $\beta_2$ and $\xi$ are the Lagrange multipliers associated
with constraints $\bf C2$, $\bf C3$ and $\bf C6$ respectively, and where
$\theta=\left(\beta_1,\beta_2,Q_1^B,Q_1^C,\xi\right)$. Here, $d_{\theta}(.)$ is
the function defined by~\eqref{eq:separating_curve_fini}.
\end{lemma}
The uniqueness of the above global solution can be proved using arguments
similar to those of the proof of Proposition~1 in~\cite{tsp1}. Note that due to
the above lemma, there is \emph{at most} one user in each subset
$\mathcal{L}_i^A$ who is likely to modulate both protected and non protected
subcarriers. If such a ``pivot-user'' exists, then it is necessarily located
\emph{on} the curve $d_{\theta}(.)$. Therefore, there are at most $I^A$
pivot-users in cell~$A$.

\subsubsection{Step 2: From Single Cell to Multicell Resource Allocation}

We now consider the problem of joint resource allocation
(Problem~\ref{prob:2d_multi}) while still assuming that users of each cell are
aligned on equispaced parallel lines. Recall the definition of
$\mathcal{L}_i^c$ as the subset of users of cell~$c$ located on line~$i$ ($i=1
\ldots I^c$). The following lemma implies that any optimal solution to
Problem~\ref{prob:2d_multi} has in each cell the same form as the solution to
the single cell problem given by Lemma~\ref{lem:single_2d}.
\begin{lemma}
\label{lem:multi_2d}
Assume that the positions of users of each cell~$c\in\{A,B,C\}$ are aligned on
$I^c$ parallel equispaced lines. Any global solution
$\{\gamma\ku^c,P\ku^c,\gamma\kd^c,P\kd^c\}_{\stackrel{c=A,B,C}{k=1\ldots K^c}}$
to Problem~\ref{prob:2d_multi} satisfies the following. Let
$Q_1^c=\sum_{k=1}^{K^c}\gamma\ku^c P\ku^c$ designate the power transmitted
by base station~$c$ on the reused subcarriers~$\cal I$. There exist nine
positive numbers $\{\beta_1^c$, $\beta_2^c$, $\xi^c\}_{c=A,B,C}$ such
that~\eqref{eq:allocinf_2d}, \eqref{eq:allocsup_2d} hold in each cell.
\end{lemma}
The proof of Lemma~\ref{lem:multi_2d} is provided in
Appendix~\ref{app:lem_multi}. For each cell~$c\in\{A,B,C\}$, denote by $\bar{c}$
and $\bar{\bar{c}}$ the other two cells and recall the definition of function
$d_{\theta}(x)$ given by~\eqref{eq:separating_curve_fini} for any $x\in[-D,D]$
and $\theta\in\mathbb{R}_+^5$. Lemma~\ref{lem:multi_2d} states that when an
optimal solution to Problem~\ref{prob:2d_multi} is applied, then there exists in
each cell~$c=A,B,C$ a curve $d_{\theta^c}(.)$, where  
$\theta^c=(\beta_1^c,\beta_2^c,Q_1^{\bar{c}},Q_1^{\bar{\bar{c}}},\xi^c)$, that
separates protected users from non protected users.
\subsubsection{Step 3: Asymptotic Performance of the Optimal Resource
Allocation}

Denote by $\theta\cK= \big(\beta_1\cK$, $\beta_2\cK$, $Q_1^{\bar{c},(K)}$,
$Q_1^{\bar{\bar{c}},(K)}$, $\xi\cK\Big)$ for $c=A,B,C$ any set of parameters
chosen such that Lemma~\ref{lem:multi_2d} holds. Superscript $(K)$ is used in
order to stress the dependency of the above parameters on the number of users
$K$. We now characterize the behaviour of $\theta\cK$ as the number $K$ of users
tends to infinity. Once the behaviour of $\theta\cK$ determined, the asymptotic
behaviour of both the separating curves $d_{\theta\cK}(.)$ and the total
transmit power of the optimal solution to Problem~\ref{prob:2d_multi} can be
fully characterized. Assume that the total number~$K$ of users tends to infinity
in such a way that $K^c/K\to 1/3$ \emph{i.e.,} the number of users in each cell
is asymptotically equivalent. Denote by~$B$ the total bandwidth of the system.
Define $r_k$ as the target rate of user~$k$ in nats/s \emph{i.e.,} $r_k=B R_k$
where $R_k$ is the data rate requirement of user~$k$ in nats/s/Hz. Since the sum
$\sum_k r_k$ of rate requirements tends to infinity, we let the bandwidth $B$
grow to infinity and we assume that $K/B \to t$ where $t$ is a positive real
number. We use in the sequel the notation $I\cK$ to designate the number of
parallel equispaced lines in cell~$c$. We also assume that $I\cK$ is such that
\begin{eqnarray*} 
I\cK\xrightarrow[K\to\infty]{}\infty\:, &
\frac{I\cK}{K}\xrightarrow[K\to\infty]{} 0\:.
\end{eqnarray*}
In order to simplify the proof of the results, we assume without restriction
that the rate requirement $r_k$ for each user~$k$ is upper-bounded by a certain
constant $r_{\max}$ where $r_{\max}$ can be chosen as large as needed. We also
assume that for each user~$k$, $y_k\geq \epsilon$ where $\epsilon>0$ can be
chosen as small as needed.

As a matter of fact, sequences $\beta_1\cK$, $\beta_2\cK$, $Q_1\cK$, $Q_2\cK$,
$\xi\cK$ are upper-bounded (refer to Appendix~E in~\cite{phd} for the proof).
One can thus extract convergent subsequences from the above sequences. With a
slight abuse of notation, $\theta\cK=\Big(\beta_1\cK$, $\beta_2\cK$,
$Q_1^{\bar{c},(K)}$, $Q_1^{\bar{\bar{c}},(K)}$, $\xi\cK\Big)$ will designate
from now on these convergent subsequences and their respective limits will be
denoted by $\theta^c=(\beta_1^c$, $\beta_2^c$, $Q_1^{\bar{c}}$,
$Q_1^{\bar{\bar{c}}}$, $\xi^c)$. We now provide a system of equation satisfied
by the accumulation points $\theta^c=(\beta_1^c$, $\beta_2^c$ $Q_1^{\bar c}$,
$Q_1^{\bar{\bar{c}}}$, $\xi^c)$. Due to Lemma~\ref{lem:multi_2d}, the power
$Q_1\cK=\sum_{k=1}^{K^c}\gamma\ku^c P\ku^c$ transmitted by base station~$c$
on the non protected subcarriers~$\cal I$ can be written as
\begin{equation}\label{eq:finite_Q1_2d}
Q_1\cK= \sum_{\stackrel{k\in\{1\ldots K^c\}}{y_k<
d_{\theta\cK}(x_k)}} R_k {\cal F}
\left(x_k,y_k,\beta_1\cK,Q_1^{\bar{c},(K)},Q_1^{\bar{\bar{c}},(K)},
\xi\cK\right)+
\sum_{\stackrel{k\in\{1\ldots K^c\}}{y_k=d_{\theta\cK}(x_k)}}
\gamma_{k,1}^c P_{k,1}^c\:,
\end{equation}
where function $\cal F$ is defined as
\begin{equation}\label{eq:calF_2d}
{\cal F}(x,y,\beta,{\cal Q}',{\cal Q}'',\xi)=\frac{f^{-1}\left(
\frac{g_1(x,y,{\cal Q}',{\cal Q}'')}{1+\xi}\beta\right)}
{g_1(x,y,{\cal Q}'',{\cal Q}'')
C\left(\frac{g_1(x,y,{\cal Q}',{\cal Q}'')}{1+\xi}\beta\right)}
\end{equation}
for each $\left(x,y,\beta,{\cal Q}',{\cal Q}'',\xi\right)$ $\in
[-D,-D]\times [\epsilon,D]\times \mathbb{R}_+^4$. While the first term
in~\eqref{eq:finite_Q1_2d} represents the power allocated to all the users of
cell~$c$ that are uniquely assigned non protected subcarrier from subset
$\mathcal{I}$, the second term in the same equation represents the power
transmitted to the (at most) $I\cK$ pivot-users in the same subset. Since we
assume that $I\cK / K \to 0$, one can show~\cite{phd} that the latter term
is negligible with respect to the first term and that it tends to zero as $K \to
\infty$. Thus, it will be denoted in the sequel by $o_K(1)$,
where $o_K(1)$ stands for any term that converges to zero as $K$ tends to
infinity. Define for each cell~$c=A,B,C$ the following measure $\nu\cK$ on the
Borel sets of $\mathbb{R}_+\times \mathbb{R}\times \mathbb{R}_+$ as
$$
\nu\cK(I,J,L)
=\frac{1}{K^c}\sum_{k=1}^{K^c}\delta_{(r_k,x_k,y_k)}(I,J,L)\:,
$$ 
where $I,J,L$ are intervals of $\mathbb{R}_+$, $\mathbb{R}$ and $\mathbb{R}_+$
respectively and where $\delta_{(r_k,x_k,y_k)}$ is the Dirac measure at point
$(r_k,x_k,y_k)$ \emph{i.e.,} $\delta_{(r_k,x_k,y_k)}(I,J,L)=1$ if $r_k\in I$,
$x_k \in J$, $y_k \in L$ and $\delta_{(r_k,x_k,y_k)}(I,J,L)=0$ otherwise. Note
that $\nu\cK(I,J,L)$ can be interpreted as the number of users of cell~$c$ whose
rate requirement in nats/s is inside $I$, whose x-coordinate is inside~$J$ and
whose y-coordinate is inside~$L$, normalized by $K^c$. In other words, measure
$\nu\cK$ characterizes both the geographical distribution of users in cell~$c$
and their attribution to the different rate requirements. Replacing $R_k$
(in nats/s/Hz) by $\frac{r_k\mbox{ (nats/s) }}{B}$ in~\eqref{eq:finite_Q1_2d},
we obtain 
\begin{equation}
\begin{multlined}
Q_1\cK=
\frac{K^c}{B}\int_{0}^{r_{\max}}\int_{-D}^{D}
\int_{\max\{|x|/\sqrt{3},\epsilon\}}^{d_{\theta\cK}(x)} r {\cal F}
\left(x,y,\beta_1\cK,Q_1^{\bar{c},(K)},Q_1^{\bar{\bar{c}},(K)},\xi\cK\right)
d\nu\cK(r,x,y)+ o_K(1)\:,
\end{multlined}
\label{eq:fonctionMesure_2d}
\end{equation} 
In the sequel, we assume that the following holds.
\begin{assumption}
\label{ass:measure_convergence_2d}
As $K$ tends to infinity, measure~$\nu\cK$ converges weakly to a measure
$\nu^c$. Moreover, $\nu^c$ is the measure product of a limit rate distribution
$\zeta^c$ times a limit location distribution~$\lambda^c$. Finally, $\lambda^c$
is absolutely continuous with respect to the Lebesgue measure on $\mathbb{R}^2$.
\end{assumption}
Note that given the definition of $\zeta^c$, the value $\bar{R}^c$ defined as 
\begin{equation}\label{eq:bar_r}
\bar{R}^c=\frac{t}{3}\int_{0}^{r_{\max}}r\: d\zeta^c(r)
\end{equation}
represents the total average rate requirement per channel use in cell~$c$.
Here, recall that $t$ is the limit of $K/B$ as $K\to\infty$. It is intuitive
that~$Q_1\cK$ as given by~\eqref{eq:fonctionMesure_2d} converges in this case to
a constant $Q_1^c$ defined by
\begin{equation}\label{eq:Q1d_2}
Q_1^c = {\bar R}^c \int_{x=-D}^{D}
\int_{y=\max\{|x|/\sqrt{3},\epsilon\}}^{d_{\theta^c}(x)}
{\cal F}\left(x,y,\beta_1^c,Q_1^{\bar{c}},Q_1^{\bar{\bar{c}}},\xi^c\right)\:
d\lambda^c(x,y)\:.
\end{equation}
Using the same approach as above and recalling that $g_2(x,y)=g_1(x,y,0,0)$, one
can show that the power~$Q_2\cK$ transmitted by base station~$c$ on the
protected subcarriers~$\mathcal{P}_c$ converges as $K\to\infty$ to
\begin{equation}\label{eq:Q2d_2d}
Q_2^c = {\bar R}^c \int_{x=-D}^{D}
\int_{y=d_{\theta^c}(x)}^{\frac{2D-|x|}{\sqrt{3}}} 
{\cal F}\left(x,y,\beta_2^c,0,0,0\right)\:d\lambda^c(x,y)\:.
\end{equation}
Now recall the expression of $\gamma\ku^c$ given by Lemma~\ref{lem:multi_2d}
for all users~$k$ satisfying $y_k<d_{\theta\cK}(x_k)$. 
Plugging the latter expression into constraint $\textrm{\bf C2:}
\sum_{k=1}^{K^c}\gamma\ku^c=\alpha$ of Problem~\ref{prob:2d_multi}, we obtain
\begin{equation}\label{eq:finite_gamma1_2d}
\frac{1}{B}\sum_{\stackrel{k\in\{1\ldots K^c\}}
{y_k<d_{\theta\cK}(x_k)}} r_k 
{\cal G}
\left(x_k,y_k,\beta_1\cK,Q_1^{\bar{c},(K)},Q_1^{\bar{\bar{c}},(K)},
\xi\cK\right)
+\sum_{\stackrel{k\in\{1\ldots K^c\}}
{y_k=d_{\theta\cK}(x_k)}}\gamma_{k,1}^c=\alpha\:,
\end{equation}
where we defined 
\begin{equation}\label{eq:cal_G}
{\cal G}(x,y,\beta,{\cal Q}',{\cal Q}'',\xi)=
\frac{1}{C\left(\frac{g_1(x,y,{\cal Q}',{\cal Q}'')}{1+\xi}\beta\right)}
\end{equation}
for each positive $x$, $y$, $\beta$, ${\cal Q}'$, ${\cal Q}''$ and $\xi$. 
It is thus quite intuitive that equation~\eqref{eq:finite_gamma1_2d} leads as
$K\to\infty$ to
\begin{equation}\label{eq:gamma1d_2}
{\bar R}^c \int_{-D}^{D}
\int_{y=\max\{|x|/\sqrt{3},\epsilon\}}^{d_{\theta^c}(x)}
{\cal G} (x,y,\beta_1^c,Q_1^{\bar{c}},Q_1^{\bar{\bar{c}}},\xi^c)\:
d\lambda^c(x,y)=\alpha\:.
\end{equation}
Similarly, we can show that constraint $\textrm{\bf C3:}
\sum_{k=1}^{K^c}\gamma\kd^c=\frac{1-\alpha}{3}$ of Problem~\ref{prob:2d_multi}
leads as $K\to\infty$ to
\begin{equation}\label{eq:gamma2d_2}
{\bar R}^c\int_{-D}^{D}
\int_{d_{\theta^c}(x)}^{\frac{2D-|x|}{\sqrt{3}}}
{\cal G}(x,y,\beta_2^c,0,0,0)\:d\lambda^c(x,y)=\frac{1-\alpha}{3} \:.
\end{equation}
\begin{remark}
Equations~\SinfTwo~characterize the asymptotic behaviour of $\beta_1\cK$,
$\beta_2\cK$, $Q_1\cK$, $\xi\cK$ in the case where users of each cell are
aligned on parallel equispaced lines. The generalization to the case of an
arbitrary setting of users is not straightforward, since
Lemma~\ref{lem:multi_2d} does not necessarily hold in this general case.
Nonetheless, the lemma given below states that sequences $\beta_1\cK$,
$\beta_2\cK$, $Q_1\cK$, $\xi\cK$ have the same asymptotic behaviour as given
by~\SinfTwo~even if users are not aligned on parallel lines. The proof of this
lemma relies on the following approach. We define in each cell a set of parallel
equispaced lines similar to the lines in
Figure~\ref{fig:finite_separating_curve}. We next consider the projection of
users positions on these lines using two distinct projection rules. This way,
we are able to exploit equations~\SinfTwo~to solve the two resulting
optimization problems. If the number of the latter lines is well chosen, 
then the perturbation of the location of each user will also be small. The
optimization problem can therefore be interpreted as a perturbed version of
the initial problem. The next step is to demonstrate that this perturbation of
the initial setting of users does not alter the accumulation points of sequences
$\beta_1\cK$, $\beta_2\cK$, $Q_1\cK$, $\xi\cK$. This can be done by properly
selecting the way the number of lines scales with~$K$.
\end{remark}
\begin{lemma}
\label{lem:asymptotic_2d}
Assume that $K=K^A+K^B+K^C\to \infty$ in such a way that $K/B\to t > 0$ and
$K^c/K\to 1/3$ for $c=A, B, C$. The total power
$Q_T\K=\sum_{c=A,B,C}\sum_{k=1}^{K^c}(\gamma\ku^c P\ku^c+\gamma\kd^c P\kd^c)$ 
of any optimal solution to Problem~\ref{prob:2d_multi} converges to a constant
$Q_T$. The limit $Q_T$ has the following form:
\begin{align}
Q_T=\sum_{c=A,B,C} {\bar R}^c\Bigg(
&\int_{-D}^{D}\int_{\max\{|x|/\sqrt{3},\epsilon\}}^{d_{\theta^c}(x)}
{\cal F}
(x,y,\beta_1^c,Q_1^{\bar{c}},Q_1^{\bar{\bar{c}}},\xi^c)\:
d\lambda^c(x,y)+\nonumber\\ 
&\int_{-D}^{D}\int_{d_{\theta^c}(x)}^{\frac{2D-|x|}{\sqrt{3}}}
{\cal F}(x,y,\beta_2^c,0,0,0)\:d\lambda^c(x,y)\Bigg)\:,
\label{eq:limitQT_2d}
\end{align}
where
$\theta^c=(\beta_1^c,\beta_2^c,Q_1^{\bar{c}},Q_1^{\bar{\bar{c}}},\xi^c)$
and where for each $c=A,B,C$, the system of equation~\SinfTwo~is satisfied in
variables $\theta^c$, $Q_1^c$. Here, $(x,\theta)\mapsto d_{\theta}(x)$ is the
function defined by~\eqref{eq:separating_curve_fini}.\\
\indent Moreover, for each $c=A,B,C$ and for any arbitrary fixed value
$(Q_1^A,Q_1^B,Q_1^C)=$ $(\tilde{Q}_1^A$, $\tilde{Q}_1^B$, $\tilde{Q}_1^C)$, the
system of equation \SinfTwo~admits at most one solution $(\tilde{\beta}_1^c$,
$\tilde{\beta}_2^c$, $\tilde{\xi}^c)$.
\end{lemma}
Lemma~\ref{lem:asymptotic_2d} states that the limit $Q_T$ of the total
transmit power can be computed once we have found a set of parameters
$\{\beta_1^c,\beta_2^c,Q_1^c,\xi^c\}_{c=A,B,C}$ that satisfy~\SinfTwo~in the
three cells $c=A,B,C$. However, these twelve parameters are
underdetermined by this system of nine equations. We are
nonetheless capable of finding $\{\beta_1^c,\beta_2^c,Q_1^c,\xi^c\}_{c=A,B,C}$
such that the above lemma holds. This can be done thanks to the fact that $Q_T$
is the limit of the transmit power of an \emph{optimal} solution to the joint
resource allocation problem. Therefore,
$\{\beta_1^c,\beta_2^c,Q_1^c,\xi^c\}_{c=A,B,C}$ can be chosen as
any set of parameters that satisfy the system of equation~\SinfTwo~in the
three cells~$A$, $B$, $C$ \emph{and} for which the total power $Q_T$ as given
by~\eqref{eq:limitQT_2d} is minimal. To that end, we propose
Algorithm~\ref{algo:asym_curves_2d} which performs an exhaustive search w.r.t
points $(Q_1^A,Q_1^B,Q_1^C)$ inside a certain search interval. In practice, the
set of points $(Q_1^A,Q_1^B,Q_1^C)$ probed by the above algorithm can be
determined by resorting to numerical methods. 
\subsection{Selection of Curves $\{d_{\textrm{subopt}}^c(.)\}_{c=A,B,C}$}
\label{sec:selection_d}

We now proceed to the relevant determination of the separating curves
$\{d_{\textrm{subopt}}^c(.)\}_{c=A,B,C}$ associated with the proposed allocation
algorithm (Algorithm~\ref{algo:suboptimal_2d}). We propose to set
$d_{\textrm{subopt}}^c(x)$ such that
$$
\forall x\in[-D,D],\: d_{\textrm{subopt}}^c(x)=d_{\theta^c}(x)\:, c=A,B,C\:.
$$
where $(\theta^c)_{c=A,B,C}$ is the output of
Algorithm~\ref{algo:asym_curves_2d} and where $(x,\theta)\mapsto d_{\theta}(x)$
is the function defined by~\eqref{eq:separating_curve_fini}.
\begin{remark}
Note that the asymptotic separating curves $d_{\theta^c}(.)$ do not depend on
the particular configuration of the cells, but rather on an asymptotic
description of the network \emph{i.e.,} on the average rate requirement
$\bar{R}^c$ and on the asymptotic distribution~$\lambda^c$ of users.
\end{remark} 
\begin{remark}
\label{rem:complexity}
Curves $d_{\theta^c}(.)$ can be set before the base stations are brought into
operation. They can also be updated once in a while if~$\bar{R}^c$
or~$\lambda^c$ are subject to changes. However, since such changes are typically
slow, computational complexity of determining $d_{\theta^c}(.)$ is not a major
issue. 
\end{remark}
\subsection{Asymptotic Optimality of the Proposed Algorithm}
\label{sec:asym_opt_2d}

Denote by $Q_{\textrm{subopt}}\K$ the total transmit power of 
Algorithm~\ref{algo:suboptimal_2d} in the case where the separating curves
$\{d_{\textrm{subopt}}^c(.)\}_{c=A,B,C}$, are selected using
Algorithm~\ref{algo:asym_curves_2d}. Recall the definition of
$Q_T\K$ as the total transmit power of an optimal solution to the multicell
resource allocation problem (Problem~\ref{prob:2d_multi}). The following theorem
states that Algorithm~\ref{algo:suboptimal_2d} is asymptotically optimal. Its
proof is provided in~\cite{phd}.
\begin{theo}
\label{the:subopt_2d}
Assume that the separating curves $\{d_{\textrm{subopt}}^c(.)\}_{c=A,B,C}$
are set such that $d_{\textrm{subopt}}^c(x)=d_{\theta^c}(x)$ for all
$x\in[-D,D]$, where $(\theta^c)_{c=A,B,C}$ is the output of
Algorithm~\ref{algo:asym_curves_2d} and where $(x,\theta)\mapsto d_{\theta}(x)$
is the function defined by~\eqref{eq:separating_curve_fini} for any $x\in[-D,D]$
and $\theta\in\mathbb{R}_+^5$.
The following equality holds:
\begin{equation*}
\lim_{K\to\infty} Q_{\textrm{subopt}}\K = \lim_{K\to\infty} Q_T\K=Q_T\:,
\end{equation*}
where $Q_T$ is the constant defined by Lemma~\ref{lem:asymptotic_2d}.
\end{theo}
Note that the above theorem implies that $Q_{\textrm{subopt}}\K$ is bounded, at
least for sufficiently large $K$. This means that there exists
$K_0\in\mathbb{N}^*$ such that Problem~\ref{prob:opt_multi_2d} is feasible for
all $K\geq K_0$ (refer to Remark~\ref{rem:feasibility}).
\subsection{Selection of the Best Reuse Factor}
\label{sec:alpha}

During the cellular network design process, the selection of a relevant value
of $\alpha$ allowing to optimize the network performance is of crucial
importance. In practice, the reuse factor should be fixed prior to resource
allocation and it should be independent of the particular cells configuration.
Recall the definition of $Q_T\K=Q_T\K(\alpha)$ as the total transmit power
associated with an optimal solution to the resource allocation problem. We
define the optimal reuse factor as the value $\alpha_{\textrm{opt}}$ that
minimizes the asymptotic transmit power
$Q_T(\alpha)=\lim_{K\to\infty}Q_T\K(\alpha)$ (given by
Lemma~\ref{lem:asymptotic_2d})
\emph{i.e.,}
\begin{equation}\label{eq:alpha_opt}
\alpha_{\textrm{opt}}=\arg\min_{\alpha\in[0,1]}Q_T\:.
\end{equation}
In practice, $\alpha_{\textrm{opt}}$ can be obtained by computing $Q_T(\alpha)$
for different values of $\alpha$ in a grid. Note that computational complexity
is not an issue here (refer to Remark~\ref{rem:complexity}).
\section{Numerical Results}
\label{sec:simus}

In our simulations, we considered the classical ``free space propagation model''
with a carrier frequency $f_0=2.4 GHz$. Path loss in dB of user~$k$ in cell~$c$
($c=A,B,C$) is thus given by $\rho_k (dB) = 20\log_{10}(|ck|) + 100.04$, where
$|ck|$ stands for the distance between user~$k$ and base station~$c$. The
thermal noise power spectral density is equal to $N_{0} = -170$ dBm/Hz. Denote
by~$S$ the surface of any of the considered sectors of cells~$A,B,C$. Each one
of these three sectors is assumed to have the same uniform asymptotic
distribution $\lambda$ of users, where $d\lambda(x,y)= dx dy / S$. The average
rate requirement $\bar{R}^c$ in bits/s/Hz (defined in nats/sec/Hz
by~\eqref{eq:bar_r}) is assumed to be the same in each cell: $\bar{R}^A=$
$\bar{R}^B=$ $\bar{R}^C=$ $\bar{R}$.

\noindent {\bf Selection of the reuse factor}

\noindent In Figure~\ref{fig:asym_alpha_2d}, we plot $\alpha_{\textrm{opt}}$
defined by~\eqref{eq:alpha_opt} for different values of the average rate
$\bar{R}$. As expected, $\alpha_{\textrm{opt}}$ is decreasing with respect to
$\bar{R}$. Indeed, the larger the value $\bar{R}$, the higher the level of
interference, and the greater the number of users that should be assigned
protected subcarriers.

\noindent {\bf Selection of separating curves
$\boldsymbol{\{d_{\textrm{subopt}}^c(.)\}_{c=A,B,C}}$}

\noindent Once the reuse factor is set to the value $\alpha_{\textrm{opt}}$, the
separating curves $\{d_{\textrm{subopt}}^c(.)\}_{c=A,B,C}$ associated with the
proposed suboptimal allocation algorithm should be chosen to be equal to the
asymptotic optimal curves $\{d_{\theta^c}(.)\}_{c=A,B,C}$ given by
Lemma~\ref{lem:asymptotic_2d}. Since we are considering the case where the
asymptotic distribution of users is the same in the three sectors,
Algorithm~\ref{algo:asym_curves_2d} yielded in all our simulations three
identical separating curves \emph{i.e.,} for all $x\in[-D,D]$,
$d_{\textrm{subopt}}^A(x)=d_{\textrm{subopt}}^B(x)=d_{\textrm{subopt}}^C(x)$.
Figure~\ref{fig:asym_d_2d} plots $d_{\textrm{subopt}}^A(.)$ for different values
of the average rate $\bar{R}$. 

\noindent {\bf Performance of the proposed allocation algorithm}

\noindent From now on, the positions of users in each sector are assumed to be
uniformly distributed random variables. We also assume that all users have the
same target rate, and that $K^A=K^B=K^C$. In the sequel,
$r_T=\sum_{k=1}^{K^c}r_k$ designates the sum rate per sector measured in bits/s.
Let us study the performance of the proposed allocation algorithm
(Algorithm~\ref{algo:suboptimal_2d}) in the case where the separating curves
$\{d_{\textrm{subopt}}^c(.)\}_{c=A,B,C}$ are selected as in
Subsection~\ref{sec:selection_d} (see Figure~\ref{fig:asym_d_2d}). 

We first validate the asymptotic optimality of
Algorithm~\ref{algo:suboptimal_2d}. To that end, we consider 5 values of the
number~$K$ of users comprised between 30 and 300. For each one of these values,
the system bandwidth~$B=B(K)$ is chosen such that $K/B=t=15\times 10^{-6}$. For
example, the bandwidth is equal to~5 MHz when $K=75$ \emph{i.e.,} when
$K^A=K^B=K^C=25$. This way, the number of users increases in accordance with the
description of the asymptotic regime given earlier in
Section~\ref{sec:optimal_asyp_analysis_2d_aligned}. Next, we compute the
transmit powers $Q_{\textrm{subopt}}\K$ spent when
Algorithm~\ref{algo:suboptimal_2d} is applied for a large number of realizations
of the random positions of users. We finally evaluate the associated mean value
$\mathbb{E}\left[Q_{\textrm{subopt}}\K\right]$ (expectation is taken w.r.t the
random positions of users) and compare it with the asymptotic \emph{optimal}
transmit power $Q_T=\lim_{K\to\infty}Q_T\K$ as given by
Lemma~\ref{lem:asymptotic_2d}. The results of this comparison are illustrated in
Figure~\ref{fig:asym_error_2d}. Note that the difference between
$Q_{\textrm{subopt}}\K$ and~$Q_T$ decreases with the number of users. This
difference can be considered negligible even for a moderate number of users
equal to 50 per sector. This sustains that the proposed allocation algorithm is
asymptotically optimal.

From now on, the system bandwidth~$B$ is equal to~5 MHz and the number of users
per sector is fixed to 25. In Figure~\ref{fig:sub_opt_2d}, we compare the
proposed algorithm with the allocation scheme introduced
in~\cite{papandriopoulos}. In the latter work, the authors set the value of the
reuse factor~$\alpha$ to one \emph{i.e.,} all the available subcarriers are
reused in all the cells. The resource allocation problem they address consists
(as in our paper) in minimizing the total power that should be spent by the
network in order to achieve all users' rate requirements. In this context, they
propose a distributed iterative allocation algorithm similar to
Algorithm~\ref{algo:ping_pong_2d}. The main difference is that, while
Algorithm~\ref{algo:ping_pong_2d} is only applied to a subset
$\mathcal{K}_I^A\cup\mathcal{K}_I^B\cup\mathcal{K}_I^C$ of users, the algorithm
of~\cite{papandriopoulos} is applied to all the users in the network. As a
matter of fact, this difference has no significant effect on the computational
complexity of the scheme of~\cite{papandriopoulos}, which is also of order
$O(N_{\textrm{iter}} N_{\textrm{grad}} K)$ as 
Algorithm~\ref{algo:ping_pong_2d} (see Subsection~\ref{sec:complexity}).
However, Figure~\ref{fig:sub_opt_2d} shows that for all the different values of
the sum rate $r_T$, considerable gains can be achieved by applying our resource
allocation scheme instead of that of~\cite{papandriopoulos} without any
additional computational complexity.

\noindent {\bf Convergence rate of Algorithm~\ref{algo:ping_pong_2d}}

\noindent We plot in Figure~\ref{fig:convergence} the number $N_{\textrm{iter}}$
of iterations of Algorithm~\ref{algo:ping_pong_2d} as a function of the required
accuracy \emph{i.e.,} the maximum relative change in the transmit powers
$Q_1^c$ ($c=A,B,C$) from iteration to another beyond which convergence of the
algorithm is achieved. Figure~\ref{fig:convergence} shows that
Algorithm~\ref{algo:ping_pong_2d} converges quickly within a very good accuracy
even for a sum rate as high as 9 Mbps.

\noindent {\bf Performance of the proposed allocation algorithm in the discrete
case}

\noindent We now address the so-called \emph{discrete} case where the sharing
factors should be integer multiples of $1/N$. In this context, we propose the
following approach to compute the resource allocation parameters. We first
apply Algorithm~\ref{algo:suboptimal_2d} to obtain the continuous-valued sharing
factors $\gamma\ku^c$ and $\gamma\kd^c$ for $c=A,B,C$. Next, we round the number
of assigned subcarriers $\gamma\ku^c N$ and $\gamma\kd^c N$ to the nearest
smaller integer. In order to compensate for the slight decrease of each sharing
factor due to rounding, the power allocated to each user should be slightly
increased so as to keep the same achievable rate. To that end, the power
allocation should be recomputed (this time, keeping fixed sharing factors).
This can be achieved by straightforward adaptation of
Algorithms~\ref{algo:ping_pong_2d} and~\ref{algo:protected_2d}.
In Figure~\ref{fig:discrete}, we plot the required transmit power in the 
discrete case assuming the following values of the total number~$N$ of
subcarriers: $N= 72, 192, 360$ as recommended in WiMax~\cite{wimax_book}.
Figure~\ref{fig:discrete} shows that our allocation algorithm continues to do
relatively well even after rounding the sharing factors, provided that the
total number of subcarriers is moderately large enough.

\noindent {\bf Performance of the proposed allocation algorithm in larger
networks}

\noindent We now turn our attention to the 21-sector network of
Figure~\ref{fig:12_BS}. This network is composed of 7 duplicates of the 3-sector
system of Figure~\ref{fig:3cells_2D_model}. Note from Figure~\ref{fig:12_BS}
that the subcarriers of subsets $\mathcal{P}_A$, $\mathcal{P}_B$,
$\mathcal{P}_C$ are no more interference-free. However, the number of
interferers for users modulating in these subsets is always smaller than the
number of interferers for users modulating in subset~$\mathcal{I}$.

In this context, we propose the following procedure. We first fix the separating
curve $\{d_{\textrm{subopt}}^c(.)\}$ in each sector~$c$ and the reuse factor
$\alpha$ to the values given by Sections~\ref{sec:selection_d}
and~\ref{sec:alpha} respectively \emph{i.e.,} as if the network were composed of
only three sectors. Note that Algorithm~\ref{algo:protected_2d} cannot be
applied anymore since the users outside the curves
$\{d_{\textrm{subopt}}^c(.)\}$ are now subject to multicell interference.
Instead, we apply a straightforward adaptation of
Algorithm~\ref{algo:ping_pong_2d} to the case of more than three sectors. In
Figure~\ref{fig:sub_opt_12}, we plot both the transmit power of the
above proposed algorithm and that of the distributed and iterative allocation
scheme of~\cite{papandriopoulos} (both averaged w.r.t the random positions of
users) assuming a 21-sector setting. We note from Figure~\ref{fig:sub_opt_12}
that, while the scheme of~\cite{papandriopoulos} fails to converge for sum
rates~$r_T$ larger or equal to 9 Mbps, our allocation algorithm converges in all
the considered cases. It furthermore results in considerably smaller transmit
powers. However, comparing Figures~\ref{fig:sub_opt_2d} and~\ref{fig:sub_opt_12}
reveals that the transmit power of the proposed
algorithm is significantly larger in the 21-sector setting than in the 3-sector
setting. Reducing this gap requires a large amount of research and is out of
the scope of this paper.

\section{Conclusions}

In this paper, we addressed the problem of resource allocation for the downlink
of a sectorized OFDMA network assuming fractional frequency reuse and
statistical CSI. In this context, we proposed a practical resource allocation
algorithm that can be implemented in a distributed manner. The proposed
algorithm divides users of each cell into two groups which are geographically
separated by a fixed curve: Users of the first group are constrained to
interference-free subcarriers, while users of the second are constrained to
subcarriers subject to interference. If the aforementioned separating curves are
relevantly chosen, then the transmit power of this simple algorithm tends, as
the number of users grows to infinity, to the same limit as the minimal power
required to satisfy all users' rate requirements. Therefore, the simple scheme
consisting in separating users beforehand into protected and non protected users
is asymptotically optimal. This scheme is frequently used in cellular systems,
but it has never been proved optimal in any sense to the best of our knowledge.
Finally, we proposed a method to select a relevant value of the reuse factor.
The determination of this factor is of great importance for the dimensioning of
wireless networks.
\appendices
%
\section{Proof of Lemma~\ref{lem:multi_2d}}
\label{app:lem_multi}

\noindent {\sl Notations.}  
In the sequel, $\mathbf{x}_{ABC}$ represents a vector of multicell allocation
parameters such that 
${\bf x}_{ABC}=[{{\bf x}_A}^T,{{\bf x}_B}^T,{{\bf x}_B}^T]^T$ where
$\mathbf{x}_A=[(\mathbf{P}^A)^T,(\boldsymbol{\gamma}^A)^T]^T$, 
$\mathbf{x}_B=[(\mathbf{P}^B)^T,(\boldsymbol{\gamma}^B)^T]^T$ and 
$\mathbf{x}_C=[(\mathbf{P}^C)^T,(\boldsymbol{\gamma}^C)^T]^T$,
and where for each $c=A,B,C$, $\mathbf{P}^c=[P_{1,1}^c,P_{1,2}^c,\ldots,
P_{K^c,1}^c,P_{K^c,2}^c]^{T}$ and
$\boldsymbol{\gamma}=[\gamma_{1,1}^c,\gamma_{1,2}^c,\ldots$, $\gamma_{K^c,1}^c$,
$\gamma_{K^c,2}^c]^{T}$. We respectively denote by
$Q_1(\mathbf{x}_c)=\sum_k \gamma\ku^c P\ku^c$ and $Q_2(\mathbf{x}_c)=\sum_k
\gamma\kd^c P\kd^c$ the powers transmitted by base station~$c$ in the
interference subset~$\cal I$ and in the protected subset~${\cal P}_c$. When
resource allocation ${\bf x}_{ABC}$ is used, the total power transmitted by the
network is equal to 
$Q({\bf x}_{ABC})=\sum_c Q_1(\mathbf{x}_c)+Q_2(\mathbf{x}_c)$.
Recall that Problem~\ref{prob:2d_multi} is nonconvex. It cannot be solved using
classical convex optimization methods. Denote by 
${\bf x}_{ABC}^*=[{{\bf x}_A^*}^T,{{\bf x}_B^*}^T,{{\bf x}_C^*}^T]^T$ any global
solution to Problem~\ref{prob:2d_multi}.
\smallskip

\noindent {\bf Characterizing ${\bf x}_{ABC}^*$ via single cell results}.

From ${\bf x}_{ABC}^*$ we construct a new vector ${\bf x}_{ABC}$ which is as
well a global solution and which admits a ``binary'' form: for each cell~$c$,
$\gamma\ku^c=0$ if $y_k>d_{\theta_c}(x_k)$ and $\gamma\kd^c=0$ if
$y_k<d_{\theta_c}(x_k)$, for a certain curve $d_{\theta_c}(.)$. For cell~$A$,
vector ${\bf x}_A$ is defined as a global solution to the \emph{single cell}
Problem~\ref{prob:single_2d} when 
\begin{itemize}
\item[a)] the admissible nuisance constraint $\cal Q$ is set to 
${\cal Q}=Q_1({\bf x}_A^*)$,
\item[b)] the gain-to-interference-plus-noise-ratio in subset~$\cal I$ is set
to $g\ku=g\ku\left( Q_1({\bf x}_{B}^*), Q_1({\bf x}_{C}^*) \right)$.
\end{itemize}
Vectors ${\bf x}_B$ and ${\bf x}_C$ are defined similarly, by replacing $A$ by
$B$ or $C$ in the above definition. Denote by 
${\bf x}_{ABC} =[{{\bf x}_A}^T,{{\bf x}_B}^T,{{\bf x}_C}^T]^T$ the
allocation obtained by the above procedure. The following claim holds.
\begin{claim}
Resource allocation parameters ${\bf x}_{ABC}$ and ${\bf x}_{ABC}^*$ coincide:
${\bf x}_{ABC}={\bf x}_{ABC}^*$.
\end{claim}
\begin{proof}
It is straightforward to show that ${\bf x}_{ABC}$ is a feasible point for the
joint multicell problem (Problem~\ref{prob:2d_multi}) in the sense that
constraints $\bf C1$-$\bf C4$ of Problem~\ref{prob:2d_multi} are met. This is
the consequence of the low nuisance constraint $Q_1({\bf x}_c) \leq Q_1({\bf
x}_c^*)$ which ensures that the interference which is \emph{produced} by each
base station when using the new allocation ${\bf x}_{ABC}$ is no bigger than
the interference produced when the initial allocation ${\bf x}_{ABC}^*$ is used.
Second, it is straightforward to show that ${\bf x}_{ABC}$ is a global solution
to the multicell problem (Problem~\ref{prob:2d_multi}). Indeed, the power
$Q_1(\mathbf{x}_c)+Q_2(\mathbf{x}_c)$ spent by base station~$c$ is necessarily
less than the initial power $Q_1(\mathbf{x}_c^*)+Q_2(\mathbf{x}_c^*)$ \emph{by
definition} of the minimization Problem~\ref{prob:single_2d}. Thus 
$Q({\bf x}_{ABC})\leq Q({\bf x}_{ABC}^*)$. Of course, as ${\bf x}_{ABC}^*$ has
been chosen itself as a global minimum of $Q$, the latter inequality should hold
with equality: $Q({\bf x}_{ABC})= Q({\bf x}_{ABC}^*)$. Therefore, 
${\bf x}_{ABC}^*$ and ${\bf x}_{ABC}$ are both global solutions to the multicell
problem (Problem~\ref{prob:2d_multi}). As an immediate consequence, inequality
$Q_1(\mathbf{x}_c)+Q_2(\mathbf{x}_c)\leq
Q_1(\mathbf{x}_c^*)+Q_2(\mathbf{x}_c^*)$ holds with equality in all the three
cells~$c=A,B,C$:
\begin{equation}
Q_1(\mathbf{x}_c)+Q_2(\mathbf{x}_c)= Q_1(\mathbf{x}_c^*)+Q_2(\mathbf{x}_c^*)\:.
\label{eq:xetoilemin}
\end{equation}
Clearly, ${\bf x}_A^*$ is a feasible point for Problem~\ref{prob:single_2d} when
setting ${\cal Q} = Q_1({\bf x}_A^*)$ and  $g\ku = g\ku\big(Q_1({\bf x}_B^*)$,
$Q_1({\bf x}_C^*)\big)$. Indeed constraint $\bf C6$ is equivalent to $Q_1({\bf
x}_A^*) \leq {\cal Q}$ and is trivially met (with equality) by definition of
$\cal Q$. Since the objective function $Q_1(\mathbf{x}_A^*)+Q_2(\mathbf{x}_A^*)$
coincides with the global minimum as indicated by~(\ref{eq:xetoilemin}), ${\bf
x}_A^*$ is a global minimum for the single cell Problem~\ref{prob:single_2d}. By
Lemma~\ref{lem:single_2d}, this problem admits a unique global minimum ${\bf
x}_A$. Therefore, ${\bf x}_A^*={\bf x}_A$. By similar arguments, 
${\bf x}_B^*={\bf x}_B$ and ${\bf x}_C^*={\bf x}_C$.
\end{proof}
We thus conclude that any global solution ${\bf x}_{ABC}^*$ to the multicell
Problem~\ref{prob:2d_multi} satisfies equations~\eqref{eq:allocinf_2d},
\eqref{eq:allocsup_2d}, where $g\ku$ in the latter equations coincide with
$g\ku=g\ku\left(Q_1({\bf x}_{\bar c}^*), Q_1({\bf x}_{\bar{\bar{c}}}^*)\right)$,
and where for each cell~$c\in\{A,B,C\}$, $\bar{c}$ and $\bar{\bar{c}}$ denote
the other two cells. The proof of Lemma~\ref{lem:multi_2d} is thus complete.

\newpage

\begin{figure}[h]
\centering
\includegraphics[width=7cm]{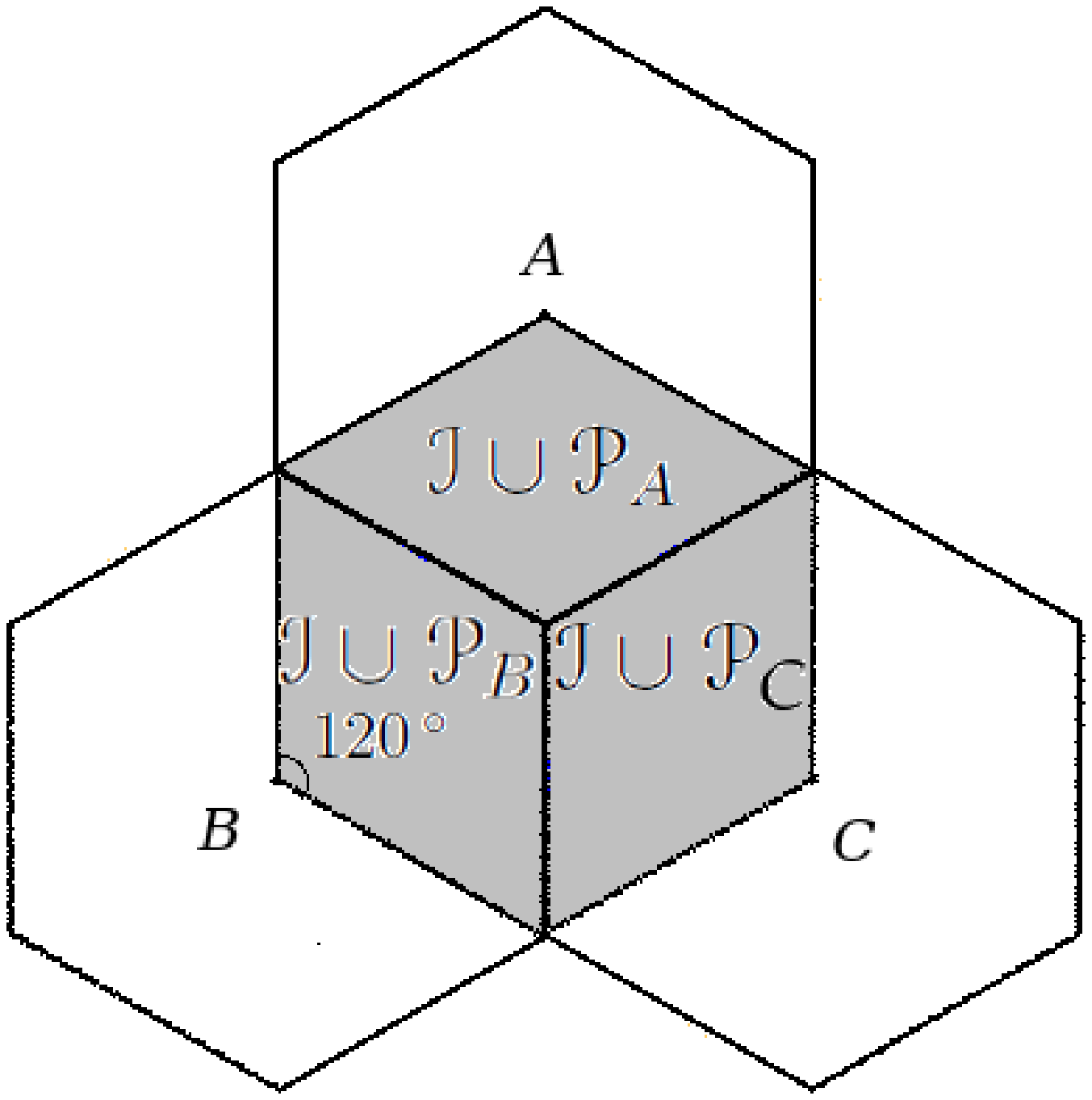}
\caption{3-cells system model and the frequency reuse scheme}
\label{fig:3cells_2D_model}
\end{figure}

\begin{figure}[h]
\centering
\includegraphics[width=12cm]{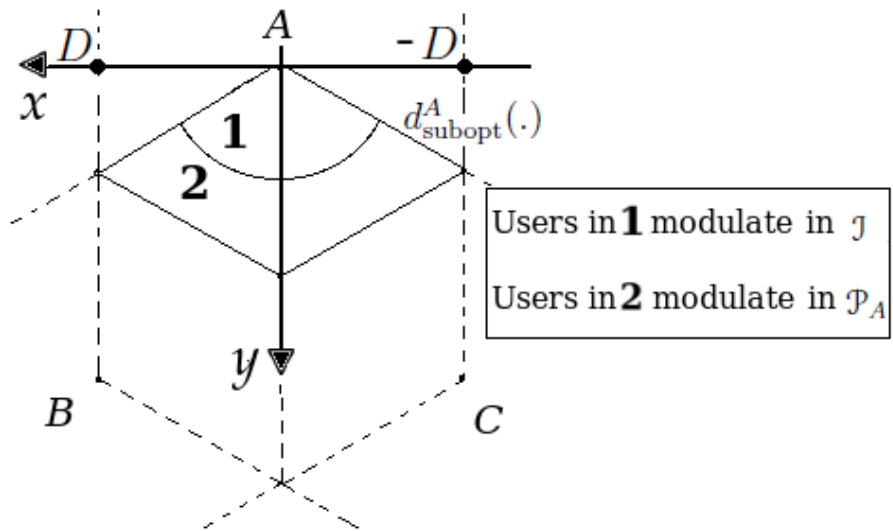}
\caption{Fixed separating curve in cell $A$}
\label{fig:subopt_2d}
\end{figure}

\begin{algorithm}
\caption{Ping-pong algorithm for three interfering cells}
\label{algo:ping_pong_2d}
\begin{algorithmic}
\STATE {\bf Initialization:} $Q_1^A \leftarrow 0$, $Q_1^B \leftarrow 0$, 
$Q_1^C \leftarrow 0$
\REPEAT
\STATE $(\beta_1^A,Q_1^A) \leftarrow$
Solve~(\ref{eq:simpa_2d})-(\ref{eq:simpb_2d}) for $c=A$
\STATE $(\beta_1^B,Q_1^B) \leftarrow$
Solve~(\ref{eq:simpa_2d})-(\ref{eq:simpb_2d}) for $c=B$
\STATE $(\beta_1^C,Q_1^C) \leftarrow$
Solve~(\ref{eq:simpa_2d})-(\ref{eq:simpb_2d}) for $c=C$
\UNTIL{convergence}
\FORALL {$c=A,B,C$}
\STATE $\{\gamma\ku^c,P\ku^c\}_{k\in\mathcal{K}_I^c}\leftarrow$
\eqref{eq:resSimpa_2d}-\eqref{eq:resSimpb_2d}
\ENDFOR
\RETURN $\{\gamma\ku^c,P\ku^c\}_{c=A,B,C,\:k\in\mathcal{K}_I^c}$
\end{algorithmic}
\end{algorithm}

\begin{algorithm}
\caption{Resource allocation for protected users}
\label{algo:protected_2d}
\begin{algorithmic}
\FORALL {$c=A,B,C$}
\STATE $\beta_2^c \leftarrow$ Solve~\eqref{eq:beta_2_tilde_2d} 
\FORALL {$k\in\mathcal{K}_P^c$}
\STATE $P\kd^c \leftarrow$ \eqref{eq:power_protected_2d}
\STATE $\gamma\kd^c \leftarrow$ \eqref{eq:gamma_protected_2d} 
\ENDFOR
\ENDFOR
\RETURN $\{\gamma\kd^c,P\kd^c\}_{c=A,B,\:k\in\mathcal{K}_P^c}$
\end{algorithmic}
\end{algorithm}

\begin{algorithm}
\caption{Proposed resource allocation algorithm}
\label{algo:suboptimal_2d}
\begin{algorithmic}
\FORALL {$c=A,B,C$}
\STATE $\mathcal{K}_P^c \leftarrow \{k\in\{1 \ldots K^c\}\: |\: y_k >
d_{\textrm{subopt}}^c(x_k)\}$
\STATE $\mathcal{K}_I^c \leftarrow \{k\in\{1 \ldots K^c\}\: |\: y_k \leq
d_{\textrm{subopt}}^c(x_k)\}$
\ENDFOR
\STATE
$\{\gamma\ku^c, P\ku^c\}_{c=A,B,C,\: k\in\mathcal{K}_I^c} \leftarrow$
Algorithm~\ref{algo:ping_pong_2d}
\STATE
$\{\gamma\kd^c, P\kd^c\}_{c=A,B,c,\: k\in\mathcal{K}_P^c} \leftarrow$
Algorithm~\ref{algo:protected_2d}
\RETURN 
$\{\gamma\ku^c, P\ku^c, \gamma\kd^c, P\kd^c\}_{c=A, B, C,\: k=1 \ldots K^c}$
\end{algorithmic}
\end{algorithm}



\begin{figure}[h]
\centering
\includegraphics[width=12cm]{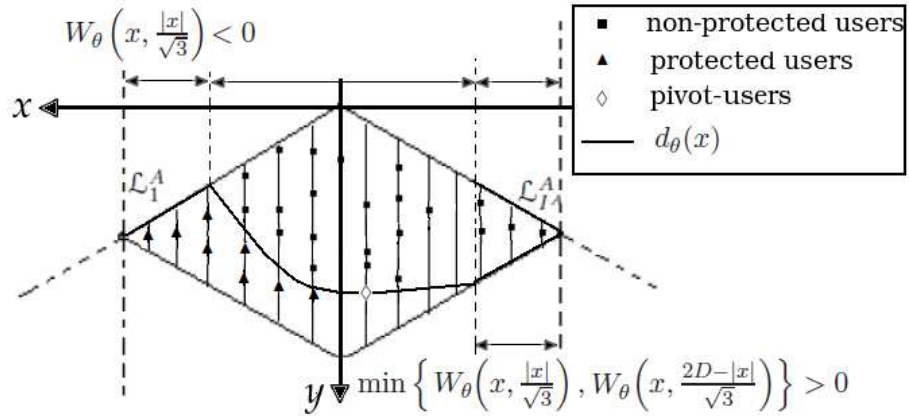}
\caption{Definition of subsets $\{\mathcal{L}_i^A\}_{i=1 \ldots I^A}$ and
of the curve $d_{\theta}(.)$}
\label{fig:finite_separating_curve}
\end{figure}

\begin{algorithm}
\caption{Determination of
$\{\theta^c\}_{c=A,B,C}$}
\label{algo:asym_curves_2d}
\begin{algorithmic}
\FORALL {$(Q_1^A,Q_1^B,Q_1^C)$}
\FOR {$c=A,B,C$}
\IF {\SinfTwo~admits a solution}
\STATE $(\beta_1^c, \beta_2^c, \xi^c)\leftarrow$
unique solution to~\SinfTwo
\STATE $\theta^c\leftarrow
(\beta_1^c,\beta_2^c,Q_1^{\bar{c}},Q_1^{\bar{\bar{c}}},\xi^c)$
\STATE 
$Q^c\leftarrow Q_1^c + {\bar R}^c \int_{-D}^{D}
\int_{d_{\theta^c}(x)}^{\frac{2D-|x|}{\sqrt{3}}}
{\cal F}(x,y,\beta_2^c,0,0,0)\:d\lambda^c(x,y)$
\ELSE
\STATE $Q^c\leftarrow\infty$
\ENDIF
\ENDFOR
\STATE
$Q_T(Q_1^A,Q_1^B,Q_1^C)\leftarrow\sum_{c=A,B,C} Q^c$
\ENDFOR
\STATE $(Q_1^{A},Q_1^{B},Q_1^{C})\leftarrow
\arg\min_{(\tilde{Q}_1^A,\tilde{Q}_1^B,\tilde{Q}_1^C)}Q_T(\tilde{Q}_1^A,
\tilde{Q}_1^B,\tilde{Q}_1^C)$
\FOR {$c=A,B,C$}
\STATE $\theta^c\leftarrow
(\beta_1^c,\beta_2^c,Q_1^{\bar{c}},Q_1^{\bar{\bar{c}}},\xi^c)$
\ENDFOR
\RETURN $\theta^A, \theta^B,\theta^C$
\end{algorithmic}
\end{algorithm}

\begin{figure}[h]
\centering
\includegraphics[width=12cm]{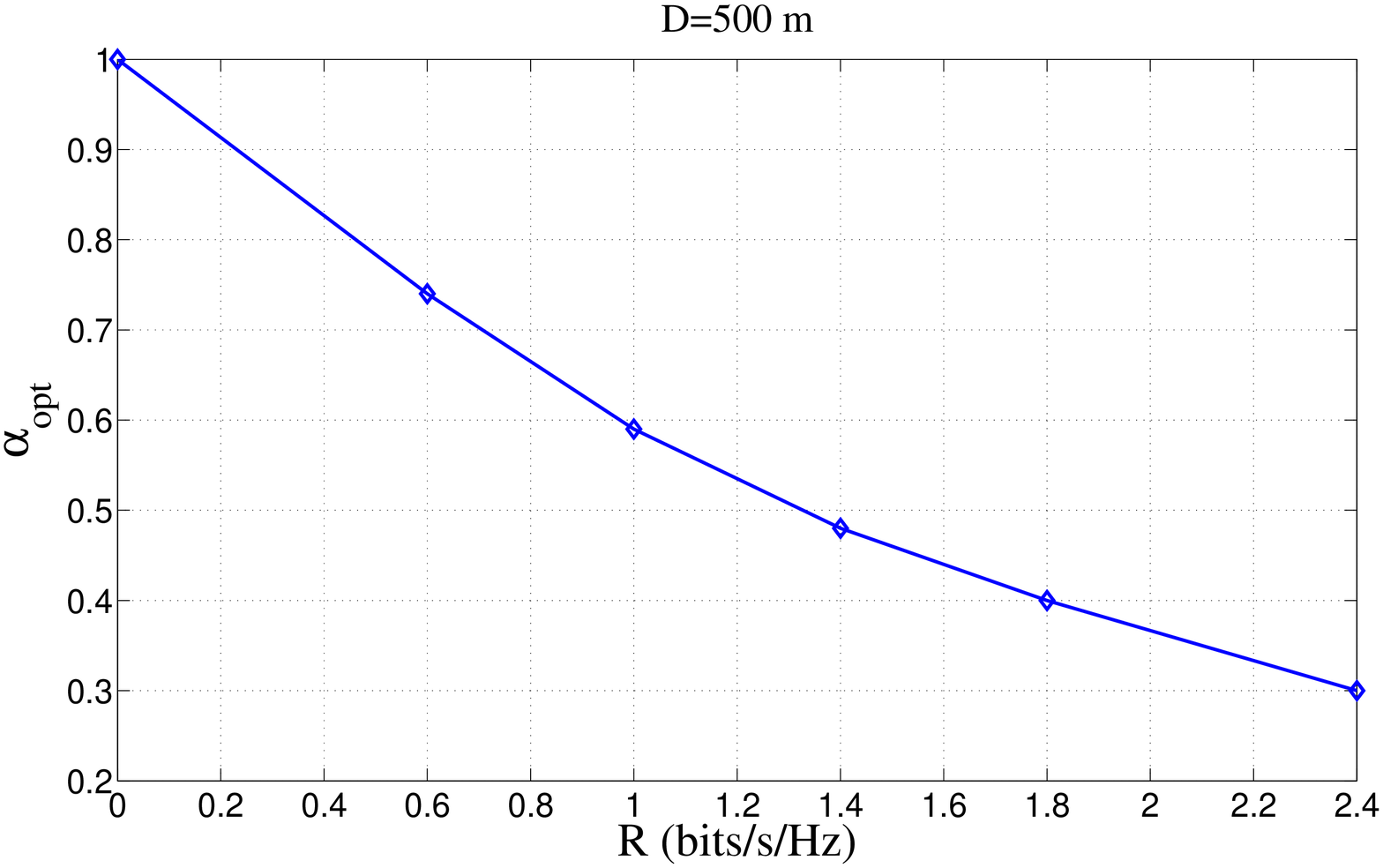}
\caption{Optimal reuse factor vs. average rate of a sector}
\label{fig:asym_alpha_2d}
\end{figure}

\begin{figure}[h]
\centering
\includegraphics[width=12cm]{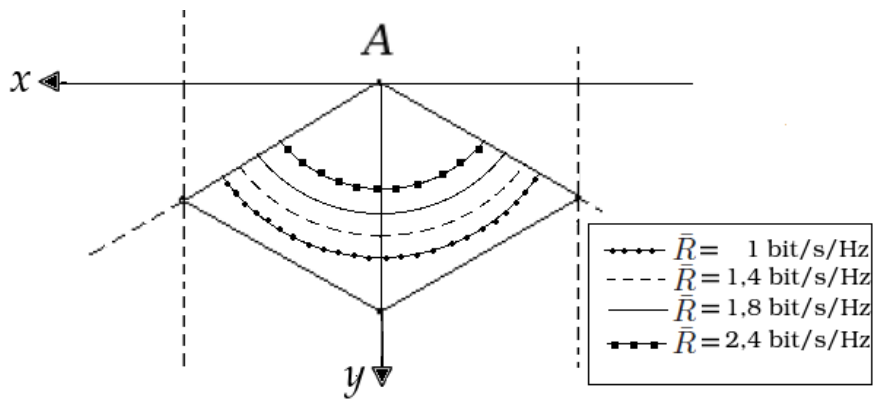}
\caption{Optimal separating curve $d_{\theta^A}(.)$}
\label{fig:asym_d_2d}
\end{figure}

\begin{figure}[h]
\centering
\includegraphics[width=12cm]{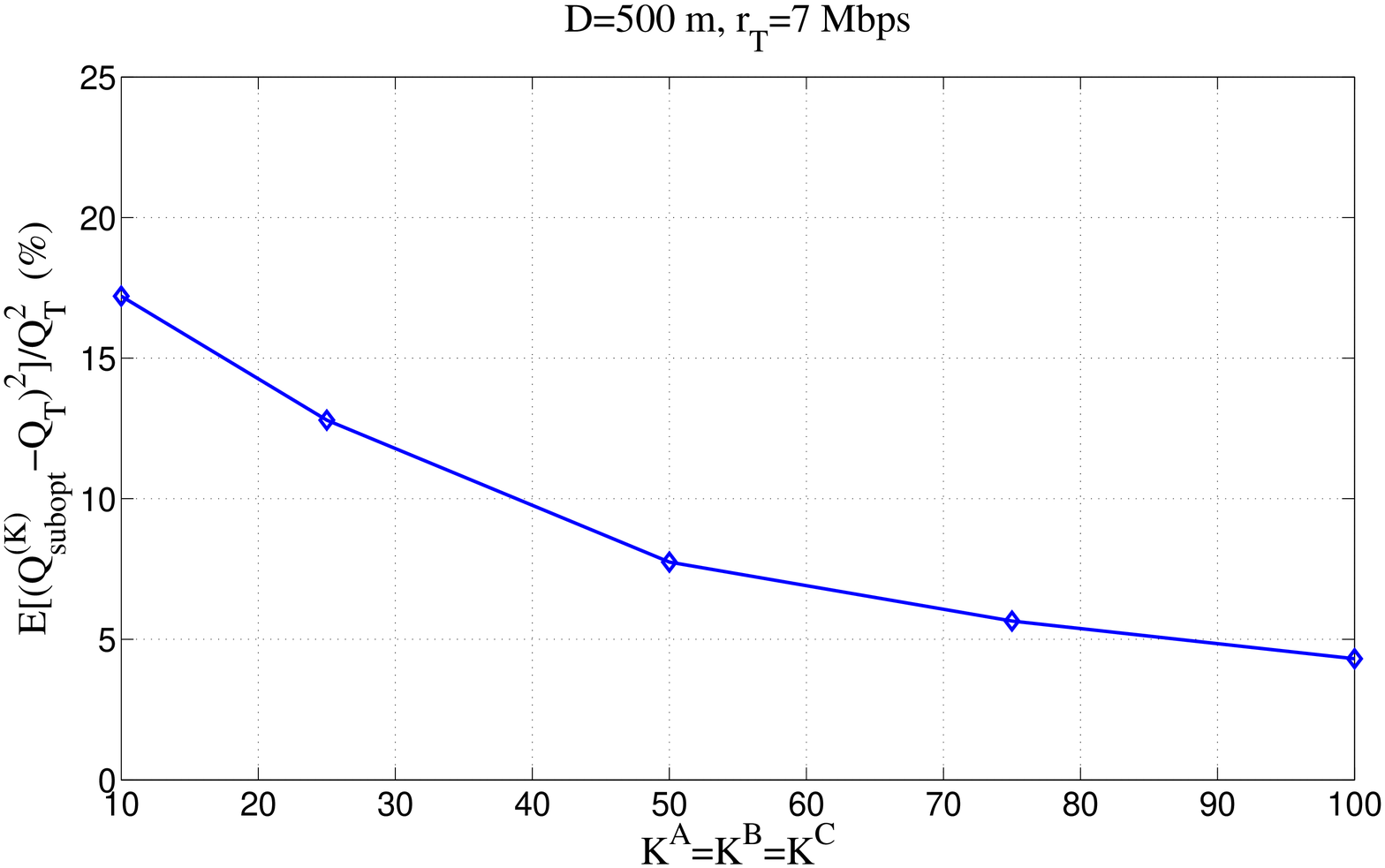}
\caption{$\mathbb{E}\left(Q_{\textrm{subopt}}\K-Q_T\right)^2/Q_T^2$
vs. number of users per sector}
\label{fig:asym_error_2d}
\end{figure}

\begin{figure}[h]
\centering
\includegraphics[width=12cm]{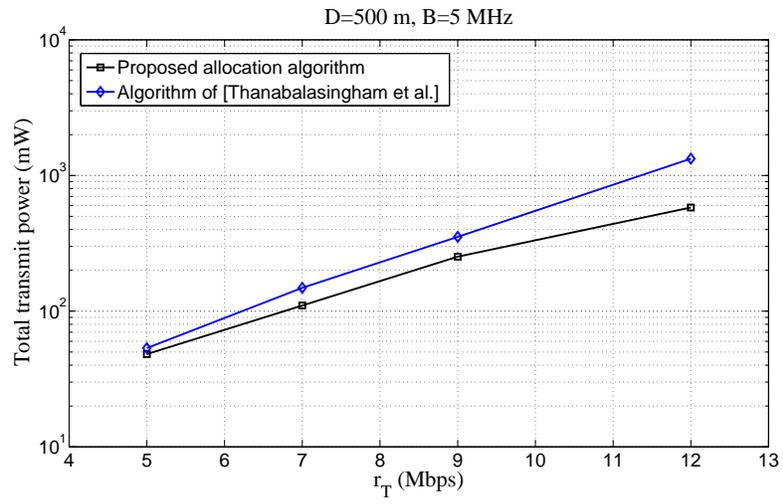}
\caption{Comparison between the proposed allocation algorithm and the
scheme of~\cite{papandriopoulos} for $K^A=K^B=K^C=25$}
\label{fig:sub_opt_2d}
\end{figure}

\begin{figure}[h]
\centering
\includegraphics[width=12cm]{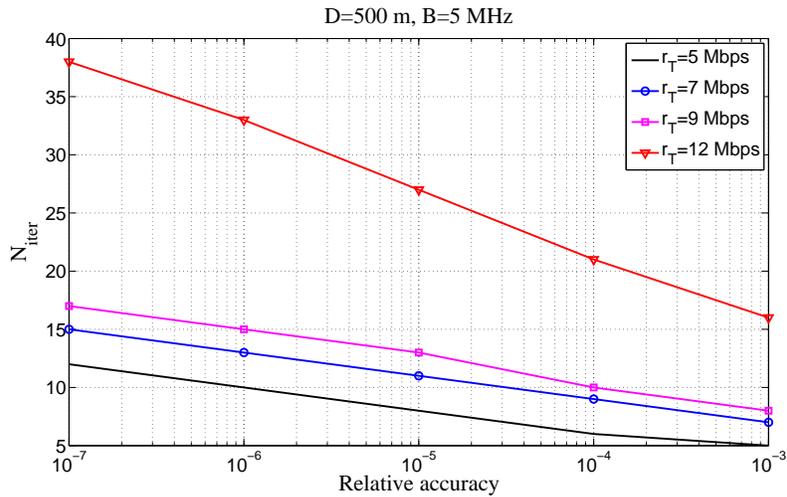}
\caption{Number of iterations of Algorithm~\ref{algo:ping_pong_2d} vs. relative
accuracy}
\label{fig:convergence}
\end{figure}
\begin{figure}[h]
\centering
\includegraphics[width=12cm]{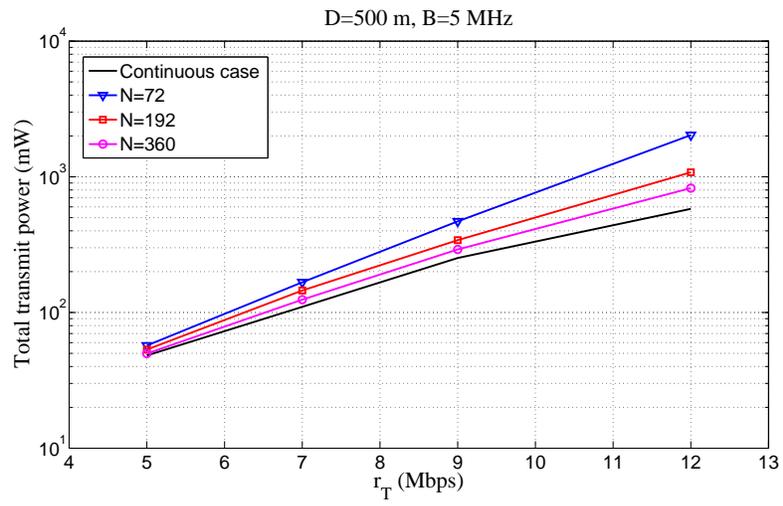}
\caption{Transmit power of the proposed algorithm in case the sharing factors
are integer multiples of $1/N$}
\label{fig:discrete}
\end{figure}

\begin{figure}[h]
\centering
\includegraphics[width=8cm]{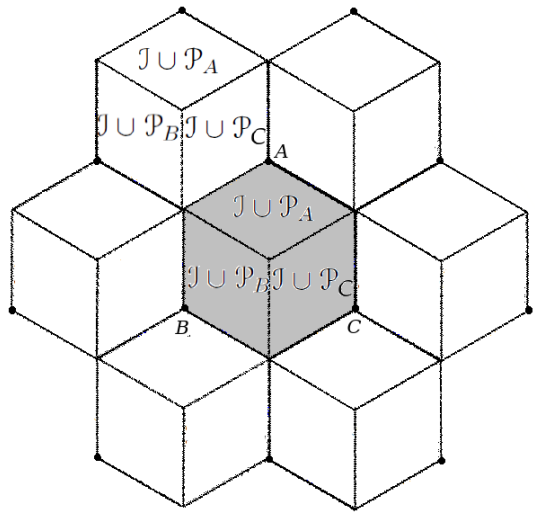}
\caption{21-sector system model and the frequency reuse scheme}
\label{fig:12_BS}
\end{figure}

\begin{figure}[h]
\centering
\includegraphics[width=12cm]{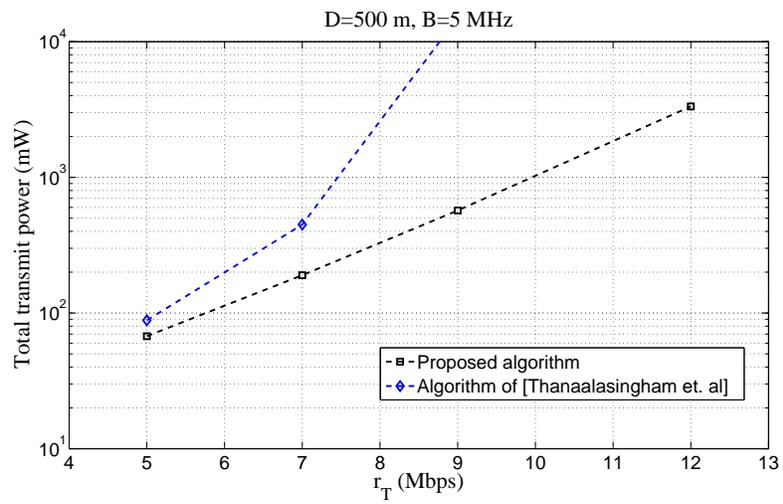}
\caption{Comparison between the proposed allocation algorithm and the
scheme of~\cite{papandriopoulos} in the case of 21 sectors for $K^c=25$}
\label{fig:sub_opt_12}
\end{figure}


\end{document}